\documentclass[11pt]{article}
\usepackage{amsmath,amssymb,amsfonts,amsthm,epsfig,verbatim}
\usepackage{times}
\usepackage{color}
\usepackage{framed}

\def\shortsodaversion{0}

\newif\ifhyper\IfFileExists{hyperref.sty}{\hypertrue}{\hyperfalse}
\hypertrue
\ifhyper\usepackage{hyperref}\fi

\newtheorem{theorem}{Theorem}

\newtheorem{lemma}[theorem]{Lemma}

\newtheorem{corollary}[theorem]{Corollary}
\newtheorem{claim}[theorem]{Claim}
\newtheorem{fact}[theorem]{Fact}

\newtheorem{definition}[theorem]{Definition}

\newcommand{\eps}{\epsilon}

\renewcommand{\tt}{\bf}

\newcommand{\Var}{\operatorname{Var}}

\newcommand{\ignore}[1]{}

\newcommand{\cref}[1]{Corollary~\ref{cor:#1}}

\newcommand{\calD}{{\cal D}}

\newcommand{\dK}{d_{{\mathrm{K}}}}
\newcommand{\bit}{{\mathrm{bit}}}

\newcommand{\R}{{\mathbb{R}}}

\newcommand{\E}{\operatorname{{\bf E}}}

\newcommand{\littlesum}{\mathop{\textstyle \sum}}

\newcommand{\poly}{\mathrm{poly}}

\newcommand{\wh}{\widehat}
\newcommand{\eqdef}{\stackrel{\textrm{def}}{=}}

\newcommand{\obj}{\mathrm{Obj}}

\renewcommand{\Pr}{\operatorname{{\bf Pr}}}

\newcommand{\opt}{\mathsf{opt}}

\newcommand{\new}[1]{#1}

\newcommand{\newad}[1]{#1}

\date{}

\begin{document}

\setcounter{page}{0}

\title{A Polynomial-time Approximation Scheme for\\ Fault-tolerant Distributed Storage}

\author{Constantinos Daskalakis\footnote{MIT. Email: \texttt{ costis@csail.mit.edu.}} \and
Anindya De\footnote{UC Berkeley. Email: \texttt{ anindya@cs.berkeley.edu.}} \and
Ilias Diakonikolas\footnote{University of Edinburgh. Email: \texttt{ilias.d@ed.ac.uk}.} \and
Ankur Moitra\footnote{Princeton and IAS. Email: \texttt{moitra@ias.edu}.} \and
Rocco A. Servedio\footnote{Columbia University. Email: \texttt{rocco@cs.columbia.edu}.}
}



\maketitle

\thispagestyle{empty}

\begin{abstract}

We consider a problem which has received considerable attention in
systems literature because of its applications to routing in delay tolerant networks and replica placement
in distributed storage systems.
In abstract terms
the problem can be stated as follows: Given a random variable $X$
generated by a known product distribution over $\{0,1\}^n$ and a target
value $0 \leq \theta \leq 1$, output a non-negative vector $w$, with
$\|w\|_1 \le 1$, which maximizes the probability of the event $w \cdot X
\ge \theta$.  This is a challenging non-convex optimization problem for
which even computing the value $\Pr[w \cdot X \ge \theta]$ of a proposed
solution vector $w$ is \#P-hard.

We provide an \new{additive} EPTAS for this problem
which, for constant-bounded product distributions,
runs in $ \poly(n) \cdot 2^{\poly(1/\eps)}$ time and outputs an
$\eps$-approximately optimal solution vector $w$ for this problem. Our
approach is inspired by, and extends,
recent structural results from the complexity-theoretic
study of linear threshold functions. Furthermore, in spite of the objective function being non-smooth, we give a \emph{unicriterion} PTAS while previous work 
for such objective functions has typically led to a \emph{bicriterion} PTAS. We believe our techniques may be applicable to get unicriterion PTAS for other non-smooth objective functions.

\end{abstract}

\newpage

\section{Introduction}

\vspace{-0.2cm}
Many applications involve designing a system that will
perform well in an uncertain environment. Sources of uncertainty include
(for example) the demand when we are designing a server, the congestion
when we are designing a routing protocol, and the failure of the
system's own components when we are designing a distributed system. Such
uncertainties are often modeled as stochastic variables, giving rise
to non-linear and non-convex optimization problems.
In this paper, we study a non-convex
stochastic optimization problem that has received considerable attention
in the systems
literature~\cite{jainetal2,fall,dimakis1,dimakis2,dimakis3,sardari} but
has remained poorly understood.

The main motivation for studying this problem comes from {\em
distributed storage}~\cite{dimakis0,bruck1,dimakis1,sardari}. The goal
in this literature is to develop methods for storing data among a set of
faulty processors in a way that makes it possible to recover the data in
its entirety despite processor failures. Clearly, to perform
this task we need to use some form of redundancy, as otherwise a single
processor failure could cause permanent loss of data. In particular,
this task contains as subproblems both the choice of an error correcting
code and the decision of how to allocate the encoded data into the
failure-prone processors, resulting in an enormous design space.

An important observation that is used throughout the literature is that
these two subproblems can be decoupled through the use of erasure codes
(see, e.g.,~\cite{erasure0, erasure1,erasure2, erasure3}). Such codes
can be used to encode the original data so that with high probability
{\em any} large enough subset of encoded data can be used to reconstruct
the original data.  In view of this observation, we can formulate the
distributed storage problem as a much simpler to state problem:

Suppose that our original data has size $\theta$~GB for some $0 \leq
\theta \leq 1$, and we use an erasure code to generate 1~GB of encoded
data. The goal is to allocate the data among $n$ failure-prone nodes so
as to maximize the probability that the original data can be recovered.
The standard formulation of the
problem~\cite{dimakis0,bruck1,dimakis1,sardari} is that each node $i$ has
some known probability $1-p_i$ of failing, and that these failures are
independent across different nodes. So, mathematically our goal is to
solve the following problem, which we call Problem (P):

\begin{framed}

\noindent {\bf Input:} 
An $n$-vector of probabilities $p = (p_1, \ldots, p_n) \in [0,1]^n$
and a parameter $\theta \in [0,1]$.

\smallskip

For $i \in [n]$ let $\mu_{p_i}$ be the distribution on $\{0, 1\}$ with
$\mu_{p_i} (1) = p_i$, and let the corresponding product distribution
over $\{0,1\}^n$ be denoted by $\calD_p = \bigotimes_{i=1}^n \mu_{p_i}$.

\smallskip

\noindent {\bf Output:} 
A weight vector $w = (w_1,
\ldots, w_n) \in \R^{n}_{\geq 0}$ satisfying $\|w\|_1 \le 1$ (such a $w$
is said to be a \emph{feasible solution}).  The goal is to maximize
$$\obj(w) \eqdef \Pr_{X \sim \calD_p} \left[ w \cdot X \ge \theta
\right].$$ A feasible solution that maximizes $\obj(w)$ is said to be an
\emph{optimal solution}. 
\noindent \new{We will denote by $\opt = \opt(p, \theta)$ the maximum value of any feasible solution.} \end{framed}


In the above formulation $w_i$ denotes the amount of data that we decide
to store in the $i$-th storage node, and $X_i$ is the indicator random
variable of the event that the $i$-th storage node does not fail.

\new{Before we proceed, we point out a connection of the optimization problem (P) above with the class of {\em Boolean halfspaces}  or
{\em Linear Threshold Functions (LTFs)} that will be crucially exploited throughout this paper. Recall that a Boolean function $f: \{0,1\}^n \to \{0,1\}$
is a halfspace if there exists a weight-vector $v \in \R^n$ and a threshold $t \in \R$ so that 
$f(x) = 1$ if and only if $v \cdot x \ge t.$ Hence, the objective function value $\obj(w)$ of a feasible weight-vector $w$ 
(i.e., $w  \in \R^{n}_{\geq 0}$ with $\|w\|_1 \le 1$) can be equivalently expressed as $\obj(w) =  \Pr_{X \sim \calD_p} \left[ h_{w, \theta}(X)=1 \right]$, 
where $h_{w, \theta}(x) = \mathbf{1}_{ \left\{x \in \{0,1\}^n : w \cdot x \ge \theta \right\}}$ is the halfspace with weight-vector $w$ and threshold $\theta$.

We remark that, even though the feasible set is continuous, it is not difficult to show that there exists a rational optimal solution. 
In particular, analogous to the linear-algebraic arguments~\cite{MTT:61, Muroga:71, Raghavan:88}, we can also show that there always exists an optimal solution 
with bit-complexity polynomially bounded in $n$; in fact, one with at most $O(n^2\log n)$ bits which is best possible~\cite{Hastad:94}. 
(As a corollary, the supremum is always attained and problem (P) is well-defined.)}

\subsection{Previous \new{and Related Work}}


\noindent {\bf Previous Work on the Problem.} 
The stochastic design problem (P) stated above was formulated explicitly in the work of Jain et al.~\cite{jainetal2}.
That work was motivated by the problem of routing in Delay Tolerant
Networks~\cite{jainetal1}. These networks are characterized by a lack of
consistent end-to-end paths, due to interruptions that may be either
planned or unplanned, and selecting routing paths is considered to be
one of the most challenging problems. The authors of~\cite{jainetal2}
reduce the route selection problem to Problem~(P) in a range of settings
of interest, and study the structure of the optimal partition as well as
its computational complexity, albeit with inconclusive theoretical
results.

One of the special cases of the problem considered in~\cite{jainetal2}
is the case where all the $p_i$'s are equal, i.e., when
$p_1=\ldots=p_n=p$. Even in this case, the structure of the optimal
solution is not well-understood. It is natural to expect that the
optimal weight vector is obtained by equally splitting
the allowed unit of weight over a subset of the indices, and setting
the weights to be $0$ on all other indices
(in other words, set $w_1=w_2=\ldots=w_k =
{1 \over k}$ and $w_{k+1}=\ldots=w_n=0$, for some $k$).
The authors of~\cite{jainetal2} consider the performance of this strategy for
different values of $p$ and $\theta$, as do the
papers \cite{dimakis1,dimakis2,dimakis3}.
Surprisingly, such partitions are not
necessarily optimal. For a counter-example, communicated to us by
R.~Kleinberg~\cite{bobby}, consider the setting where $n=5$, $\theta=5/12$
and $p=1-\epsilon$, for sufficiently small $\epsilon$. In this case, the
allocation vector $w=(1/4, 1/4, 1/6, 1/6, 1/6)$ performs better than the
uniform weight vector over any subset of the coordinates
$\{1,\ldots,5\}$. There has also been work on 
a related distributed storage
problem~\cite{sardari} that uses a slightly different model of node failures. In this model, instead of every node failing with probability $p$, a random subset of nodes of size $pn$ is assumed to fail. 
In this setting, the  conjecture
that certain symmetric allocations are optimal
is related to a conjecture of Erdos \cite{erdos} on the maximum number of
edges in a $k$-uniform hypergraph whose (fractional) matching size is
at most $s$ (see~\cite{hyper} for a detailed discussion of the connection).

\medskip
\new{
\noindent {\bf Related Work.} Stochastic optimization is an important research area with diverse applications
having its roots in the work of Dantzig~\cite{Dan55} and Beale~\cite{Beal55} that has been
extensively studied since the 1950's (see e.g.,~\cite{BL97} for a book on the topic). 
During the past couple of decades,  there has been an extensive literature on efficient approximation algorithms 
for stochastic combinatorial optimization problems in various settings, see e.g.,~\cite{KRT97, DeanGV08, BGK11, Nik10, Swamy11, LD11, LiY13} and references therein.

In many of these works, 
one wants to select a subset of (discrete independent) random variables whose sum
optimizes a certain non-linear function. For example, the objective function of our problem (P) corresponds 
to the {\em threshold probability maximization} problem~\cite{Nik10, LiY13}. 
Note that, while the solution space in the aforementioned works is typically discrete and finite in nature, 
the solution space for our problem is continuous. 
In particular, it is not always possible to discretize the space 
without losing a lot in the objective function value (see Section~\ref{ssec:techniques} for a detailed explanation of the difficulties in our setting). 

Regarding {\em threshold probability maximization}, Li and Yuan~\cite{LiY13} obtained {\em bicriterion} additive PTAS 
for stochastic versions of classical combinatorial problems, such as shortest paths, spanning trees, matchings and knapsacks. 
Roughly, they obtain a bicriterion guarantee because the function to be optimized does not have a bounded Lipschitz constant. In contrast, even though the $\Pr[w \cdot X \geq \theta]$ function that we are optimizing does not have a bounded Lipschitz constant, we are able to obtain a \emph{unicriterion} PTAS by exploiting new structural properties of near-optimal solutions that we establish in this work, as described below.
In terms of techniques, \cite{LiY13} use Poisson approximation and discretization as a main component of their results.
We note that this approach is not directly applicable in our setting, since we are dealing with a {\em weighted} sum of Bernoulli 
random variables with arbitrary {\em real} weights and we are shooting for a unicriterion PTAS.
We view the unicriterion guarantee that we achieve as an important contribution of the techniques in this work.

}



\medskip

\subsection{Our results}
\new{It is unlikely that Problem (P) can be solved exactly in polynomial time. Note that (even for the very special case when each $p_i$ equals $1/2$)
(exactly) evaluating the objective function $\obj(w)$ of a candidate solution $w$ is $\#P$-hard. (This follows by a straightforward reduction from 
the counting version of knapsack, see e.g., Theorem 2.1 of \cite{KRT} for a proof.) 
In fact, problem (P) is easily seen to lie in $NP^{\#P}$, and we are not aware of a better upper bound. We conjecture that the exact problem is intractable,
namely $\#P$-hard.}


\medskip

\new{The focus of this paper is on efficient approximation algorithms.} As our main contribution, we give an 
additive EPTAS for (P) \new{for the case that each $p_i$ is bounded away from $0$}. That is, 
we give an algorithm that for every $\eps>0$, outputs a feasible solution $w$ such that $\obj(w)$ is within an additive $\eps$ of the optimal value. 
An informal statement of our main result follows (see Theorem~\ref{thm:main} for a detailed statement):
\begin{theorem} \label{thm:main-informal} [Main Result -- informal
statement] 
\new{Fix any $\eps>0$ and} let $p=(p_1,\dots,p_n)$ be any input instance such that
$\min_i p_i \geq \eps^{\Omega(1)}$.  There is a randomized algorithm which, for any
such input vector $p$ and any input threshold $0 \leq \theta \leq 1$, runs
in $\poly(n) \cdot 2^{\poly(1/\eps)}$ time and with high probability outputs
a feasible solution vector $w$ whose value is within an additive $\eps$ of the
optimal. \end{theorem}

\subsection{Our techniques} \label{ssec:techniques}

\new{ \noindent {\bf Background.} In recent years, there has been a surge of research interest in concrete complexity theory
on various problems concerning halfspaces. These include constructions of low weight
approximations of halfspaces~\cite{Servedio:07cc,DiakonikolasServedio:09,DDFS:12stoc}, PRGs for halfspaces~\cite{DGJ+:10, MZstoc10},
property testing algorithms~\cite{MORS:10} and approximate reconstruction of halfspaces
from low-degree Fourier coefficients~\cite{OS11:chow, DDFS:12stoc} among others.

All these results use a ``structure versus randomness'' tradeoff
for halfspaces which can be described roughly as follows:  Consider
the weights of a halfspace $\mathbf{1}_{\left\{ x \in \{0,1\}^n: w \cdot x \geq \theta \right\}}$ in order of
decreasing magnitude.  If the largest-magnitude weight is ``small'' compared to the {\em 2-norm} of the weight-vector $w$,
then the Berry-Ess{\'e}en theorem (a quantitative version of the Central Limit Theorem with explicit
error bounds) implies that for independent $\{0, 1\}$ random variables 
$X_i$ (that are not too biased towards $0$ or $1$), 
the distribution of $w \cdot X$ will be well-approximated by the Gaussian distribution with the same mean and variance. 
This is a very useful statement because it implies that the discrete random variable  $w \cdot X$ essentially inherits
several nice properties of the Gaussian distribution (such as anti-concentration, strong tail bounds, and so on). 
On the other hand, if the largest-magnitude weight accounts for a
significant fraction of the $2$-norm, then the weight-vector obtained by erasing this weight has
significantly smaller $2$-norm, and we have ``made progress;'' intuitively, after a bounded number
of steps of this sort, the 2-norm of the remaining weights will be extremely small, so the
halfspace essentially depends only on the first few variables
and should be ``easy to handle'' for that reason.
These arguments can be made quantitatively precise
using the notion of the ``critical index'' (introduced in \cite{Servedio:07cc}; see
Definition \ref{def:ci}) which plays an important
role in much of the  work described above.

\smallskip

\noindent {\bf Our Contribution.} 
In this paper we show how tools from the  complexity-theoretic literature on halfspaces 
alluded to above can be leveraged in order to make algorithmic progress on our optimization problem (P).  
As we will explain below, several non-trivial technical issues arise in the context of problem (P) which require careful treatment.

At a high-level, in this work we adapt and enhance this technical machinery in order to obtain a structural understanding of the problem,
which is  then combined with  algorithmic and probabilistic techniques to obtain a PTAS.
Very roughly, we proceed as follows: We partition the space of {\em optimal} solution vectors $v^{\ast}$ into a constant number of subsets, 
based on the value of the critical index of $v^{\ast}$.
For each subset we apply a (different) algorithm which outputs a candidate (feasible) solution which is guaranteed 
to be $\eps$-optimal, {\em assuming $v^{\ast}$ belongs to the particular subset.} Since at least one subset contains an optimal solution, the best candidate
solution will be $\eps$-approximately optimal as desired. 

Of course, we need to explain how to compute a candidate solution for each subset.
A basic difficulty comes from the fact that our problem
is not combinatorial. The space of feasible solutions is continuous and even though one can easily argue that there exists
a rational optimal solution with polynomially many bits, a priori we do not know anything more about its structure.
We note that a natural first approach one would think to try in this context
would be to appropriately ``discretize'' the weights (e.g., by using a geometric subdivision, etc) 
and then use dynamic programming to optimize in the discretized space.
However, it is far from clear how to show that such a naive discretization works; 
one can easily construct examples of weight vectors $w$ such that ``rounding'' the coefficients of $w$ 
to an appropriate (inverse polynomial in $n$) granularity radically changes the value of the objective function\footnote{Moreover, we note that discretization of the space followed
by standard approaches, e.g., along the lines of~\cite{ChekuriK05}, seems to inherently lead 
to {\em bicriteria} guarantees.}. 

To compute an approximately optimal solution for each case (i.e., for $v^{\ast}$ in a particular subset as described above) 
one needs a better understanding of the structure of the optimal solutions. The reason why ``rounding'' the coefficients
may substantially change the objective function value is because for certain weight vectors $w$ the random variable
$w \cdot X$ is very concentrated, i.e., it puts a substantial fraction of its probability mass in a small interval. If on the other hand, 
$w \cdot X$  is sufficiently {\em anti-concentrated}, i.e., it puts small mass on every small interval, then it is easy to argue
that ``rounding'' does not affect the objective function by a lot. Known results~\cite{TaoVu:annals} show that the anti-concentration of $w \cdot X$ depends strongly on the {\em additive structure} of $w$.
While it is hopeless to show that all feasible weight-vectors are anti-concentrated, one could hope to show that there exists
a {\em near-optimal} solution that has good anti-concentration. Essentially, this is what we do.

Our main structural theorem (Theorem~\ref{thm:large-tail}) shows that, except in degenerate cases, 
there always exists an optimal solution whose ``tail'' has sufficiently large $L_1$-norm compared with the ``head
\footnote{If the optimal weight vector only has nonzero coordinates in the $L$ coordinates in the ``head'' 
(think of $L$ as a constant -- it will depend only on $\eps$), 
then as we show we can find an optimal vector exactly in $\poly(n) \cdot 2^{\poly(L)}$ time by an enumeration-based approach.}.
We remark that, while results of a broadly similar flavor appear
in many of these previous papers (see e.g.,~\cite{Servedio:07cc,OS11:chow,DiakonikolasServedio:09}) there are 
a few crucial differences. First, the previous works compare the $L_2$ norms of the ``head'' and the ``tail''. Most importantly, 
 all previous such results consist of re-expressing the LTFs in a ``nice" form (which includes changing the value of the threshold $\theta$). 
 Indeed, the previous arguments which assert the existence of these nice forms do not control the value of the threshold as its exact value is immaterial. In contrast, for our problem the exact threshold
in comparison to the $L_1$-norm of the weight vector is a crucial
parameter. Our structural theorem   says that every
LTF has a well-structured equivalent version in which (1) the threshold stays exactly the same
relative to the \emph{$L_1$-norm} of the weights, and (2) \emph{$L_1$-norm} of the ``tail weights''
is ``large.''   Our proof of this theorem is based on {\em linear fractional programming}, 
which is novel in this context of structural results for LTFs.
Conceptually, our structural theorem serves as a ``pre-processing'' step which ensures
that the optimal weight-vector may be assumed to be well-structured; our algorithm crucially
exploits this nice structure of the optimal solution to efficiently find a near-optimal
solution.

}



\section{Preliminaries}
\vspace{-0.1cm}
\subsection{Simplifying assumptions about the problem instance}
It is clear that if $\theta=0$ or $\theta=1$ then it is trivial to output
an optimal solution; hence throughout the rest of the paper we assume that
$0 < \theta < 1.$

Without loss of generality we may make the following assumptions about the
input $(p_1,\dots,p_n)$:
\begin{enumerate}
\item [(A1)] $p_1 \geq \cdots \geq p_n$.

\item [(A2)] $p_1 < 1-\eps$ and all $p_i \in \{\eps/(4n),\dots,k\eps/(4n)\}$, where
$k\eps/(4n)$ is the largest integer multiple of $\eps/(4n)$ that is less than $1-\eps$.
For the first claim, note that if $p_1 \geq 1-\eps$ then the solution $w=(1,0,\dots,0)$
has $\Pr_{X \sim \calD_p} \left[  w \cdot X \ge \theta \right] \geq 1-\eps$
and hence $(1,0,\dots,0)$ is an $\eps$-optimal solution as desired.
For the second claim, given an input vector of arbitrary values $p' =(p'_1,\dots,p'_n)
\in [0,1-\eps)^n$, if we round the $p'_i$ values to integer multiples
of $\eps/4n$ to obtain $p=(p_1,\dots,p_n)$, then  a simple
coupling argument gives that for any event $S$, we have
$\left|\Pr_{X \sim \calD_p}[S] - \Pr_{X \sim \calD_{p'}}[S]\right|
\leq \eps/4.$
Hence for our purposes, we may assume that the initial
$p_i$ values are ``$\eps/(4n)$-granular'' as described above.

\end{enumerate}

We further make some easy observations about optimal solutions that will be useful later.  First, it is clear that there exists an
optimal solution $w$ with $\|w\|_1=1.$  (If $\|w\|_1 < 1$ then rescaling by
$\|w\|_1$ gives a new feasible solution whose value is at least as good as the original
one.)  Second, by assumption (A1) there exists
an optimal solution $w \in \R^n_+$ that satisfies $w_i \ge w_{i+1}$
for all $i \in [n-1]$.  (If $w_i < w_{i+1}$ it is easy to see that
by swapping the two values we obtain a solution whose value is
at least as good as the original one.)

\subsection{Tools from structural analysis of LTFs:
regularity and the critical index}

\begin{definition}[regularity]
Fix any real value $\tau > 0.$  We say that a vector $w = (w_1, \ldots, w_n) \in \R^n$ is \emph{$\tau$-regular}
if $\max_{i \in [n]} |w_i| \leq \tau \|w\|_2.$
A linear form $w \cdot x$ is said to be $\tau$-regular if $w$ is $\tau$-regular.
\end{definition}

Intuitively, regularity is a helpful notion because if $w$ is $\tau$-regular then the
Berry-Ess{\'e}en theorem can be used to show that for $X \sim {\calD_p}$,
the linear form $w \cdot X$ is distributed like a Gaussian (with
respect to Kolmogorov distance) up to an error of $\eta$, where
$\eta$ depends on the regularity parameter and the parameters
$p_1,\dots,p_n$
\ifnum\shortsodaversion=1
(see Corollary 19 in the full version).
\else
(see Corollary \ref{cor:BE}).
\fi

A key ingredient in our analysis is the notion of the ``critical index''
of a linear form $w \cdot x$.
The critical index was implicitly introduced and used in
\cite{Servedio:07cc} and was explicitly used in
\cite{DiakonikolasServedio:09,DGJ+:10,OS11:chow,DDFS:12stoc}
and other works.
Intuitively, the critical index of $w$ is the first index $i$ such that
from that point on, the vector $(w_i,w_{i+1},\dots,w_n)$ is regular.  A precise definition follows:

\begin{definition}[critical index] \label{def:ci}
Given a vector $w \in \R^n$ such that $|w_1| \geq \cdots \geq |w_n| > 0$,
for $k \in [n]$ we denote by $\sigma_k$ the quantity $\sqrt{\littlesum_{i=k}^n w_i^2}$.
We define the \emph{$\tau$-critical index $c(w, \tau)$ of $w$} as the smallest
index $i \in [n]$ for which $|w_i| \leq \tau \cdot \sigma_i$. If
this inequality does not hold for any $i \in [n]$, we define $c(w, \tau) = \infty$.
\end{definition}

Given a problem instance $p$ satisfying (A1) and (A2) and a value $\eps$,
we define
\begin{equation} \label{eq:L}
L=L(\eps,\gamma)=\min\{n, \Theta(1/(\eps^2 \gamma^2) \cdot (1/\gamma) \cdot
(\log 1/(\eps\gamma))\cdot(\log(1/\eps))\},
\end{equation}
where $\gamma= \min\{p_n,1-p_1\}\geq \eps/4n$.  The idea behind this choice of
$L$ is that it is the cutoff for ``having a large
$(\eps\gamma)/\new{200}$-critical index.''

\subsection{A useful structural theorem about solutions}

\ifnum\shortsodaversion=1
In Section \ref{sec:LP} we sketch a proof that given any feasible
solution, there is another feasible solution whose value is at least
as good as the original one and which has a ``heavy tail''
with respect to the $L_1$ norm (a detailed proof is given in the full version):
\else
In Section \ref{sec:LP} we prove that given any feasible
solution, there is another feasible solution whose value is at least
as good as the original one and which has a ``heavy tail''
with respect to the $L_1$ norm:
\fi

\begin{theorem} \label{thm:large-tail}
Fix $K \in [n]$, $0 < \theta < 1$, and $w_1 \geq \cdots \geq w_n \geq 0$ such that
$\sum_{i=1}^n w_i=1$.  Let $S=\{x \in \{0,1\}^n: w \cdot x \geq \theta\}.$
Then there is a vector $v=(v_1,\dots,v_n)$ such that
\begin{enumerate}

\item [(a)] $\sum_{i=1}^n v_i = 1$ and $v_1 \geq \cdots \geq v_n \geq 0$;

\item [(b)] every $x \in S$ has $v \cdot x \geq \theta$; and

\item [(c)] either $v_{K+1}=\cdots=v_n=0$ or else $\sum_{i=1}^k v_i \leq  \new{(K+2)^{(K+2)/2}} \cdot \sum_{i=K+1}^n v_i $.

\end{enumerate}
\end{theorem}
Applying Theorem \ref{thm:large-tail} with $K=L$ as defined in (\ref{eq:L}),
we get that there exists an optimal solution $v^\ast$ that satisfies
(a) and (b), and either $v^\ast_{L+1}=\cdots=v^\ast_n=0$
or else
$\sum_{i=1}^L v^{\ast}_i \leq  \new{(L+2)^{(L+2)/2}} \cdot \sum_{i=L+1}^n v^{\ast}_i$.
 {\em Throughout the paper, we fix $v^\ast$ to be \new{such an} optimal solution vector.}

\subsection{Our approach and formal statement of the main result.}
At a high level, our approach is to
consider three mutually exclusive and exhaustive cases for $v^\ast$:
\begin{itemize}
\item {\bf Case 1:}  $v^\ast$ has $v^\ast_{L+1}=\cdots=v^\ast_n=0.$
In this case we say $v^\ast$ is an \emph{$L$-junta}.
(Note that if $L = n$ then we are in this case; hence in Cases 2 and
3 we have that $L < n.$)

\item {\bf Case 2:}  $v^\ast$ is not an $L$-junta and
$c(v^\ast,\eps\gamma/\new{200})>L.$  In this case we say that
$v^\ast$ is of \emph{type $L+1$.}

\item {\bf Case 3:}  $v^\ast$ is not an $L$-junta and
$c(v^\ast,\eps \gamma/\new{200})=K$ for some $K \in \{1,\dots,L\}.$
In this case we say that $v^\ast$ is of \emph{type $K$.}
\end{itemize}
We show (see Section \ref{sec:L-junta})
that in Case~1 it is possible to \new{efficiently compute} an {\em exactly} optimal
solution. In both Cases~2 and~3 (see Sections \ref{sec:type-Lplus1} and \ref{sec:type-i}
respectively) we show that it is possible (using two
different algorithms) to efficiently construct a set of $N
\leq \poly(n,2^{\poly(L)})$ feasible solutions such that one of them (call it $w'$) has $\obj(w') \geq \opt - \eps/2.$
Running all three procedures, we \new{thus obtain} a set of $O(nN)=\poly(n,2^{\poly(L)})$
candidate solutions such that one of them (call it $\widetilde{w}$)
is guaranteed to have $\obj(\widetilde{w}) \geq \opt - \eps/2.$  From this it
is simple to obtain an $\eps$-approximate optimal solution
(see Section \ref{sec:together}).

A precise version of our main result is given below, where by $\bit(\theta)$ we denote the bit-length of
$\theta$:

\begin{theorem} \label{thm:main} [Main Result]
There is a randomized algorithm with the following performance guarantee:
It takes as input a vector of probabilities
$p=(p_1,\dots,p_n)$ satisfying (A1) and (A2), a threshold
value $0 < \theta < 1$,
and a confidence parameter $0<\delta<1.$
It runs in $\poly(n,2^{\poly(1/\eps,1/\gamma)},\bit(\theta))
\cdot \log(1/\delta)$ time, where $\gamma= \min\{p_n,1-p_1\}\geq \eps/4n$.
With probability $1-\delta$ it outputs a feasible solution $\widetilde{w}$
such that $\obj(\widetilde{w}) \geq \opt-\eps$, and
an estimate $\widetilde{\obj}(\widetilde{w})$ of $\obj(\widetilde{w})$ that satisfies
$|\widetilde{\obj}(\widetilde{w}) - \obj(\widetilde{w})| \leq \eps.$
\end{theorem}

\section{There exist well-structured optimal
solutions:  Proof of Theorem \ref{thm:large-tail}}
\label{sec:LP}

Fix $K \in [n]$, $0<\theta<1$, and $w = (w_1, \ldots, w_n)$ with $w_1 \ge w_2 \ge \ldots w_n \ge 0$ and $\sum_{i=1}^n w_i=1$.
If $w_{i} = 0$ for all $i \in [K+1, n]$ it is clear that the weight-vector $w$ satisfies conditions (a)-(c). So, we will henceforth assume that
$W_T \eqdef \sum_{i=K+1}^n w_i >0$.

We start by defining the following {\em linear--fractional program}
$(\mathcal{LFP})$ over variables $u_1, \ldots, u_K$ and $r$.
$(\mathcal{LFP})$ is defined by the following set of linear constraints:
\begin{enumerate}
\item[(i)] For all $x \in S$, it holds $\sum_{i=1}^K u_i x_i + \sum_{i=K+1}^n w_i x_i \ge r.$

\item[(ii)] For all $i \in [K-1]$, $u_i \ge u_{i+1}$; and $u_K \ge w_{K+1}.$
\end{enumerate}
The (fractional) objective function to be maximized is
\ifnum\shortsodaversion=0
$$f_0(u_1, \ldots, u_K, r) =  \frac{r}{\sum_{i=1}^K u_i + W_T}.$$
\else
$f_0(u_1, \ldots, u_K, r) =  \frac{r}{\sum_{i=1}^K u_i + W_T}.$
\fi
Observe that $(u_1, \ldots, u_K, r) = (w_1, \ldots, w_K, \theta)$ is a feasible solution, hence the maximum value of $(\mathcal{LFP})$ is at least $\theta$.

\ifnum\shortsodaversion=1
In the full version, we show how the Charnes--Cooper transformation~\cite{CC:62} can be used to
turn $(\mathcal{LFP})$ into an essentially equivalent linear program $(\mathcal{LP})$.  Apart from one degenerate case (which corresponds to the case when $v_{K+1}=\cdots=v_n=0$),
there is a bijection between the solutions of  $(\mathcal{LFP})$ and $(\mathcal{LP})$.  It is useful to have this
equivalent linear program because (as is the case for all linear programs with finite-valued optima) the linear
program achieves its optimal value at some vertex of the polytope.  Since the vertex of
the polytope is by definition the point at which a set of constraints becomes tight, we can analyze the
corresponding system of linear equations and, using bounds on the coefficients of the equations in the system,
prove Theorem \ref{thm:large-tail}.  We give details in the full version.

\else

We now proceed to turn $(\mathcal{LFP})$ into an essentially equivalent  linear program $(\mathcal{LP})$, using the standard Charnes--Cooper transformation~\cite{CC:62}.
The linear program $(\mathcal{LP})$ has variables $t$, $s_1, \ldots, s_K$ and $\delta$ and is defined by the following set of linear constraints:
\begin{enumerate}
\item[(i)] For all $x \in S$, it holds $\sum_{i=1}^K s_i x_i + \left( \sum_{i=K+1}^n w_i x_i \right) \cdot t \ge \delta.$

\item[(ii)] For all $i \in [K-1]$, $s_i \ge s_{i+1}$; and $s_K \ge w_{K+1} \cdot t.$

\item[(iii)] $\sum_{i=1}^K s_i + W_T \cdot t = 1$; and

\item[(iv)] $t \ge 0$.

\end{enumerate}
The linear objective function to be maximized is $\delta.$

The following standard claim (see e.g.~\cite{boyd2}) quantifies the relation between the two aforementioned optimization problems:
\begin{claim} \label{claim:equiv}
The optimization problems $(\mathcal{LFP})$ and $(\mathcal{LP})$ are equivalent.
\end{claim}
\begin{proof}
Let $(u_1^{\ast}, \ldots, u_K^{\ast}, r^{\ast})$ be a feasible solution to $(\mathcal{LFP})$. It is straightforward to verify
that the vector $(t^{\ast}, s_1^{\ast}, \ldots, s_K^{\ast}, \delta^{\ast})$
with $$t^{\ast} = \frac{1}{\sum_{i=1}^K u_i^{\ast} + W_T},$$
$s_i^{\ast} = t^{\ast} u_i^{\ast}$, for $i \in [K]$,  and $\delta^{\ast} = t^{\ast} r^{\ast}$ is a feasible solution to $(\mathcal{LP})$ with the same objective function value.
It follows that the linear program $(\mathcal{LP})$ is also feasible with maximum value at least $\theta$.
Moreover, the maximum value of  $(\mathcal{LP})$ is greater than or equal to the maximum value of  $(\mathcal{LFP})$.

Conversely, if $(t^{\ast}, s_1^{\ast}, \ldots, s_K^{\ast}, \delta^{\ast})$ is a feasible solution to $(\mathcal{LP})$ with $t^{\ast} \neq 0$, then $(u_1^{\ast}, \ldots, u_K^{\ast}, r^{\ast})$
with $u_i^{\ast} = s_i^{\ast} / t^{\ast}$ and $r^{\ast} = \delta^{\ast}/t^{\ast}$ is feasible for $(\mathcal{LFP})$, with the same objective function value
$$\delta^{\ast} =\frac{r^{\ast}}{\sum_{i=1}^K u^{\ast}_i + W_T} .$$
If $(t^{\ast}, s_1^{\ast}, \ldots, s_K^{\ast}, \delta^{\ast})$ is a feasible solution to $(\mathcal{LP})$ with $t^{\ast} = 0$ and
$(u_1^{\ast}, \ldots, u_K^{\ast}, r^{\ast})$ is feasible to $(\mathcal{LFP})$ then
$$(\widetilde{u_1}, \ldots, \widetilde{u_K}, \widetilde{r}) = (u_1^{\ast}, \ldots, u_K^{\ast}, r^{\ast}) + \lambda (s_1^{\ast}, \ldots, s_K^{\ast}, \delta^{\ast})$$
is feasible to $(\mathcal{LFP})$ for all $\lambda \ge 0.$ Moreover, note that
$$\lim_{\lambda \rightarrow \infty} f_0 (\widetilde{u_1}, \ldots, \widetilde{u_K}, \widetilde{r})  = \frac{\delta^{\ast}}{\sum_{i=1}^K s_i^{\ast}} = \delta^{\ast}.$$
So, we can find feasible solutions to $(\mathcal{LFP})$ with objective values arbitrarily close to the objective value of $(t^{\ast}=0, s_1^{\ast}, \ldots, s_K^{\ast}, \delta^{\ast})$.
Therefore, the maximum value of  $(\mathcal{LFP})$ is greater than or equal to the maximum value of  $(\mathcal{LP})$.

Combining the above completes the proof of the claim.
\end{proof}

We now proceed to analyze the linear program  $(\mathcal{LP})$. 
We will show that there exists a feasible solution to  $(\mathcal{LP})$
with properties that will be useful for us. \new{Note that $S$ is by definition non-empty. In particular, the all $1$'s vector belongs to $S$. Hence, 
because of constraint (iii), the optimal value $\delta^{\ast}$ of $(\mathcal{LP})$ is at most $1$ (i.e., $(\mathcal{LP})$ is bounded).} 
Consider a vertex $\mathbf{v}^{\ast}  = (t^{\ast}, s_1^{\ast}, \ldots, s_K^{\ast}, \delta^{\ast})$ of the feasible set of $(\mathcal{LP})$
maximizing the objective function $\delta$. Claim~\ref{claim:equiv} and the observation that the optimal value of $(\mathcal{LFP})$ is at least $\theta$ imply that 
$\delta^{\ast} \ge \theta$. We consider the following two cases:

\medskip

\noindent [{\bf Case I:} $t^{\ast}=0$.] In this case, we select the desired vector $v = (v_1, \ldots, v_n)$ as follows:
We set $v_i = s_i^{\ast}$ for all $i \in [K]$ and $v_i = 0$ for $i \in [K+1, n]$. Observe that condition (c) of the theorem statement is immediately satisfied.
For condition (a), we note that
constraint (ii) of  $(\mathcal{LP})$ implies that $v_i \ge v_{i+1}$ for all $i \in [n-1]$,
while constraint (iii) implies that $\sum_{i=1}^n v_i = \sum_{i=1}^K s_i^{\ast}=1$. Finally, for Condition (b) note that by constraint (i) it follows
that $\sum_{i=1}^K v_i x_i \ge \delta^{\ast} \ge \theta$. This completes the analysis of this case.

\medskip

\noindent [{\bf Case II:} $t^{\ast} \neq 0$.] In this case, we show that $t^{\ast}$ cannot be very close to $0$.
It follows from basic LP theory that the vertex $\mathbf{v}^{\ast} = (t^{\ast}, s_1^{\ast}, \ldots, s_K^{\ast}, \delta^{\ast})$
is the unique solution of a linear system  $A' \cdot \mathbf{v}^{\ast}  = b'$ obtained from a subset of tight constraints in $(\mathcal{LP})$. We record the following fact:
\begin{fact} Consider the linear program $(\mathcal{LP})$:\\
(a) All the entries of the constraint matrix $A$ are bounded from above by $\max \{1, W_T \}$. \\
(b) The constant vector $b$ has entries in $\{0, 1\}$.\\
(c) Any coefficient not associated with the variable $t$ is in $\{0,1\}$.
\end{fact}

As mentioned above $\mathbf{v}^{\ast}$ is the unique solution of a \new{$(K+2) \times (K+2)$} 
linear system $A' \cdot \mathbf{v}^{\ast}  = b'$, where $(A', b')$ is obtained from $(A, b)$ by
selecting a subset of the rows. By Cramer's rule, we have that $t^{\ast} = \det(A'_t)/\det(A')$ where $A'_t$ is obtained by replacing the column in $A'$ corresponding
to $t^{\ast}$ with the vector $b'$. Since $A'_t$ has only $0, 1$ entries, if $\det(A'_t) \not =0$, then $\det(A'_t) \ge 1$. Since we assumed that $t^{\ast} \neq 0$, we will indeed
have that $\det(A'_t) \ge 1$. Now observe that all the columns of $A'$ except the one corresponding to $t^{\ast}$ have entries bounded from above by $1$.
The column corresponding to $t$ has all its entries bounded from above by $W_T$. By Hadamard's inequality we obtain
$$|\det(A')| \le \prod_{i=1}^{\new{K+2}} \Vert A'_i \Vert_2 \le \new{(K+2)^{(K+2)/2}} \cdot W_T.$$
By combining the above we get $$t^{\ast} \ge \new{(K+2)^{-(K+2)/2}} \cdot (1/W_T).$$

We are now ready to define the vector $v = (v_1, \ldots, v_n)$. We select $v_i = s_i^{\ast}$ for $i \in [K]$ and $v_i = t^{\ast} w_i$ for $i \in [K+1, n]$.
It is easy to verify that $v$ satisfies conditions (a)-(c) of the theorem. Indeed, we use the fact that
$\mathbf{v}^{\ast} = (t^{\ast}, s_1^{\ast}, \ldots, s_K^{\ast}, \delta^{\ast})$ is feasible for $(\mathcal{LP})$.

Constraint (iii) of $(\mathcal{LP})$ yields
$\sum_{i=1}^n v_i = \sum_{i=1}^K s^{\ast}_i + t^{\ast}  \sum_{i=K+1}^n w_i   = \sum_{i=1}^K s^{\ast}_i + t^{\ast}  W_T   =1$ as desired.
Constraint (ii) similarly implies that $v_1 \ge v_2 \ge \ldots v_n \ge 0$, which establishes condition (a).

We now proceed to establish condition (b). Let $x \in S$.
We have that
$$
\sum_{i=1}^n v_i x_i - \theta  \ge\sum_{i=1}^n v_i x_i - \delta^{\ast} =   \sum_{i=1}^K s^{\ast}_i x_i + t^{\ast} \left( \sum_{i=K+1}^n w_i x_i \right) -\delta^{\ast} \ge 0
$$
where the last inequality uses constraint (i) of $(\mathcal{LP})$.

For condition (c), since $t^{\ast} \ge \new{(K+2)^{-(K+2)/2}}  \cdot (1/W_T)$, constraint (iii) of $(\mathcal{LP})$ gives
$$\sum_{i=1}^K v_i  = \sum_{i=1}^K s^{\ast}_i = 1 - t^{\ast} W_T \le 1 - \new{(K+2)^{-(K+2)/2}}.$$
Using the fact that $\sum_{i=K+1}^n v_i = t^{\ast} W_T \ge \new{(K+2)^{-(K+2)/2}} $, we conclude that 
$$\sum_{i=1}^k v_i \le \new{(K+2)^{(K+2)/2}}  \cdot \sum_{i=K+1}^n v_i$$
This completes the proof of Theorem \ref{thm:large-tail}.  \qed

\fi

\section{Case 1:  $v^\ast$ is an $L$-junta} \label{sec:L-junta}
\vspace{-0.1cm}
In this section we prove the following theorem.

\begin{theorem} \label{thm:L-junta}
There is a (deterministic) algorithm {\tt Find-Optimal-Junta} with the
following performance guarantee:
The algorithm takes as input a vector of probabilities
$p=(p_1,\dots,p_L)$ satisfying
(A1) and (A2), a threshold value $0 < \tau < 1$, and
a parameter $0 \leq W \leq 1.$
It runs in \newad{$\poly(n,2^{\poly(L)},\bit(\tau))$} time
and outputs a head vector
$w' \in \R^{L}_{\geq 0}$ such that $\sum_{i=1}^L w'_i \leq W.$
Moreover, the vector $w'$ maximizes $\Pr[w \cdot X^{(H)} \geq \tau]$
over all $w \in \R^{L}_{\geq 0}$ that have $\sum_{i=1}^L w_i  \leq W.$
\end{theorem}

Note that Theorem \ref{thm:L-junta} is somewhat more
general than we need \new{in  order} to establish the desired result in Case 1; this is
because {\tt Find-Optimal-Junta} will also be used as a component of
the algorithm for Case 2.
As a direct corollary of Theorem \ref{thm:L-junta} we get that
{\tt Find-Optimal-Junta} finds an optimal solution in Case 1:

\begin{corollary} \label{cor:L-junta}
If $v^\ast$ is an $L$-junta, then {\tt Find-Optimal-Junta}$((p_1,\dots,p_L),
\theta, 1)$ outputs a vector $w'=(w'_1,\dots,w'_L)$ such that
$(w', \new{\mathbf{0}_{n-L}}) \in \R^n_{\geq 0}$ is an optimal solution, i.e.,
$\obj((w',\new{\mathbf{0}_{n-L}}))=\opt.$
\end{corollary}


\begin{framed}
\noindent {\bf Algorithm
{\tt Find-Optimal-Junta}:}

\medskip

\noindent {\bf Input:}  vector of probabilities $(p_1,\dots,p_L)$;
threshold $0 < \tau < 1$; parameter $W>0$

\noindent {\bf Output:} vector $w' \in \R^L_{\geq 0}$ that maximizes $\Pr[w \cdot X^{(H)} \geq \tau]$
over all $w \in \R^{L}_{\geq 0}$ that have $\sum_{i=1}^L w_i \leq W$
\begin{enumerate}
\item Let ${\cal S}$ be the set of all $2^{\Theta(L^2)}$ sets $S \subseteq \{0,1\}^L$ such that
$S=\{x \in \{0,1\}^L:  u \cdot x \geq c \}$ for some $u \in \R^L, c \in \R.$

\item For each $S \in {\cal S}$, check whether the following linear program over variables $w_1,\dots,w_L$
is feasible and if so let $w^{(S)} \in \R^L$ be a feasible solution:
\vspace{-0.1cm}
\[
\text{For all~}x \in S,
w \cdot x \geq \tau; \quad
 w_1, \cdots, w_L \geq 0; \quad
w_1+\cdots+w_L \leq W.
\]
\vspace{-0.3cm}
\item For each $w^{(S)}$ obtained in the previous step, compute
$\Pr[w^{(S)} \cdot X^{(H)} \geq \tau]$
and output the vector $w^{(S)}$ for which this is the largest.
\end{enumerate}
\end{framed}


\ignore{
}

\ifnum\shortsodaversion=1

This case is rather simple.  Procedure {\tt Find-Optimal-Junta}
outputs a vector $w'=(w'_1,\dots,w'_L)$
that maximizes the desired probability over all
non-negative vectors whose coordinates sum to at most $W$.
This is done in a straightforward way, using
linear programming and an exhaustive enumeration of all
linear threshold functions that depend only on the first $L$ variables.
The key fact used in the proof \cite{Chow:61,MTT:61} is that there are only
$2^{\Theta(L^2)}$ distinct Boolean functions that can be represented as linear
threshold functions, and these can be enumerated in essentially
$2^{\Theta(L^2)}$ time.  Details are given in the full version.

\else

This case is rather simple.  Procedure {\tt Find-Optimal-Junta}
outputs a vector $w'=(w'_1,\dots,w'_L)$
that maximizes the desired probability over all
non-negative vectors whose coordinates sum to at most $W$.
This is done in a straightforward way, using
linear programming and an exhaustive enumeration of all
linear threshold functions that depend only on the first $L$ variables.

We now proceed with the proof of Theorem \ref{thm:L-junta}.
We first give the simple running time analysis.  It is well known
(see e.g.,~\cite{Chow:61}) that, as claimed in Step~1
of {\tt Find-Optimal-Junta}, there are $2^{\Theta(L^2)}$ distinct Boolean
functions over $\{0,1\}^L$ that can be represented as halfspaces
$u \cdot x \geq c$.  It is also well known
(see \cite{MTT:61}) that for every $S \in {\cal S}$, there is a
vector $u=(u_1,\dots,u_L)$ and a threshold $c$ such that
$S=\{x \in \{0,1\}^L:  u \cdot x \geq c\}$ where each $u_i$ and
$c$ is an integer of absolute value at most $2^{\Theta(L \log L)}$.
Thus it is possible to enumerate over all elements $S \in {\cal S}$
in $2^{\Theta(L^2 \log L)}$ time.  Since for each fixed $S$ the linear
program in Step~2 has $O(2^L)$ constraints over $L$ variables, the
claimed running time bound follows.

The correctness argument is equally simple.
There must be some $S \in {\cal S}$ which is precisely the set of those
$x \in \{0,1\}^L$ that maximizes
$\Pr_{X \sim \mu_{p_1} \times \cdots \times \mu_{p_L}}[w \cdot X \geq \tau]$
over all $w \in \R^{L}_{\geq 0}$ that have $\littlesum_{i=1}^L w_i \leq W.$
Step~2 will identify a feasible solution
for this $S$, and hence the vector $w'=(w'_1,\dots,w'_L)$
that {\tt Find-Optimal-Junta} outputs will achieve this maximum
probability.  This concludes the proof of Theorem \ref{thm:L-junta}.

\fi


\section{Case 2:  $v^\ast$ is type $L+1$} \label{sec:type-Lplus1}
\vspace{-0.1cm}
Recall that in Case 2 the optimal solution $v^\ast$ is not an $L$-junta,
so it satisfies
$\littlesum_{i=1}^L v^\ast_i \leq \new{(L+2)^{(L+2)/2}} \cdot \littlesum_{i=L+1}^n v^\ast_i $, 
and $c(v^\ast,\eps) > L.$ For this case we prove the following theorem:
\begin{theorem} \label{thm:type-Lplus1}
There is a (deterministic) algorithm {\tt Find-Near-Opt-Large-CI} with the
following performance guarantee:
The algorithm takes as input a vector of probabilities
$p=(p_1,\dots,p_n)$ satisfying
(A1) and (A2) and a threshold value $0 < \theta < 1.$
It runs in $\poly(n,2^{\poly(L)},\bit(\theta))$ time
and outputs a set of $N \leq \poly(n,2^{\poly(L)})$ many feasible solutions.
If $v^\ast$ is of type $L+1$ then one of the feasible solutions $w'$
that it outputs satisfies $\obj(w') \geq \opt - \eps/2.$
\end{theorem}
\subsection{Useful probabilistic tools and notation.}
\noindent {\bf Anti-concentration.}
We say that a real-valued random variable
$Z$ is \emph{$\eps$-anti-concentrated at radius $\delta$} if for every
interval of radius $\delta$, $Z$ lands in that interval with probability at most $\eps$, i.e.,
\ifnum\shortsodaversion=0
\[
\text{for all $t \in \R$,}\quad
\Pr[|Z-t|\leq\delta] \leq \eps.
\]
\else
for all $t \in \R$, $\Pr[|Z-t|\leq\delta] \leq \eps.$
\fi
\ifnum\shortsodaversion=0
We will use the following simple result, which says that anti-concentration
of a linear form under a product distribution can only improve by
adding more independent coordinates:
\begin{lemma} \label{lem:extension}
Fix $(q_1,\dots,q_n) \in [0,1]^n$ and let
$\bigotimes_{i=1}^n \mu_{q_i}$
denote the corresponding product distribution
over $\{0,1\}^n.$
Fix any weight-vector $w^{(k)} \in \R^k$ and suppose that
the random variable $w^{(k)} \cdot X^{(k)}$, where $X^{(k)} \sim \bigotimes_{i=1}^k \mu_{q_i}$,
is $\eps$-anti-concentrated at radius
$\delta$.  Then for any $w^{(n-k)} \in \R^{n-k}$, the random variable
$w \cdot X$, where $w = (w^{(k)}, w^{(n-k)} )$ and $X \sim \bigotimes_{i=1}^n \mu_{q_i}$ is also $\eps$-anti-concentrated
at radius $\delta.$
\end{lemma}

\fi
\medskip

\noindent {\bf Notation.}
Much of our analysis in this section will deal separately with the coordinates
$1,\dots,L$ and the coordinates $L+1,\dots,n$; hence the following
terminology and notation will be convenient.  For an $n$-dimensional
vector $w \in \R^n$, in this section
we refer to $(w_1,\dots,w_L)$ as the ``head''
of $w$ and we write $w^{(H)}$ to denote this vector; similarly
we write $w^{(T)}$ to denote the ``tail'' $(w_{L+1},\dots,w_n)$ of $w$.
We sometimes refer to a vector in $\R^L$ as a ``head vector'' and to a
vector in $\R^{n-L}$ as a ``tail vector.''
In a random variable $w^{(H)} \cdot X^{(H)}$ the randomness is over the
draw of $X^{(H)} \sim \bigotimes_{i=1}^L \mu_{p_i}$,
and similarly for a random variable $w^{(T)} \cdot X^{(T)}$
the randomness is over the draw of $X^{(T)}
\sim \bigotimes_{i=L+1}^n \mu_{p_i}.$

\subsection{The algorithm and its analysis.}
Case 2 is more involved than Case 1.  We first explain some of the analysis
that motivates our approach
\ifnum\shortsodaversion=0
(Lemmas \ref{lem:1} and \ref{lem:2} below)
\else
(Lemmas 12 and 13 in the full version)
\fi
and then explain how the
algorithm works (see Steps 1 and 2 of {\tt Find-Near-Opt-Large-CI}).

Let us say that a vector $w=(w_1,\dots,w_n) \in \R^n$ has a
\emph{$\kappa$-granular tail} if the following condition holds
\ifnum\shortsodaversion=1
(throughout the rest of Section \ref{sec:type-Lplus1},
$\kappa=\poly(1/n,1/2^{\poly(L)})$):
\else
(throughout the rest of Section \ref{sec:type-Lplus1},
$\kappa=\poly(1/n,1/2^{\poly(L)})$; we will specify its value more
precisely later):
\fi
\begin{itemize}
\item {\bf [$w=(w_1,\dots,w_n)$ has a $\kappa$-granular tail]:}
For $L+1 \leq i \leq n$, each coordinate $w_i$ is an integer multiple of $\kappa.$
\end{itemize}
\ifnum\shortsodaversion=1
The first stage of our analysis is to show 
that there is a feasible solution $w'$ such that both the head
and tail have some useful properties:  the tail weights are
granular and the tail random variable is sharply concentrated around its
mean, while the head random variable is anti-concentrated and moreover 
gives a high-quality solution to a problem with a related threshold.
(See Lemma 12 in the full version.)
\else
The first stage of our analysis is to show (assuming that $v^\ast$
is type $L+1$) that there is a feasible solution such that both the head
and tail have some useful properties:  the tail weights are
granular and the tail random variable is sharply concentrated around its
mean, while the head gives a high-quality solution to a
problem with a related threshold
(see condition (3) below):
\begin{lemma} \label{lem:1}
Suppose $v^\ast$ is type $L+1$.  Then there is a feasible
solution $w'=(w'_1,\dots,w'_n) \in \R^n_{\geq 0}$ such that
$w'_1 \geq \cdots \geq w'_n \geq 0$ \ignore{and $\obj(w') \geq \opt
- \eps/4$ }which satisfies the following:
\begin{enumerate}
\item The vector $w'$ has a $\kappa$-granular tail.  Hence for
$M \eqdef \poly(1/\kappa)$, there are
non-negative integers $A',B',C' \leq M$ such that
$ \sum_{i=L+1}^n (w'_i)^2 = A' \kappa^2$,
$\sum_{i=L+1}^n w'_i p_i = B' \kappa (\eps/(4n)),$ and
$\sum_{i=L+1}^n w'_i = C' \kappa.$

\item Let $\mu'$ denote $\E[w'^{(T)} \cdot X^{(T)}]$, i.e., $\mu' = B' \kappa
(\eps/(4n)).$
The random variable $w'^{(T)} \cdot X^{(T)}$
is strongly concentrated around its
mean:
\begin{equation} \label{eq:tailconc}
\Pr[|w'^{(T)} \cdot X^{(T)} - \mu'| \geq
\sqrt{A' \cdot \ln(200/\eps)}\cdot
\newad{\kappa}
] \leq \eps/100.
\end{equation}

\item The head random variable $w'^{(H)} \cdot X^{(H)}$ satisfies
\begin{equation} \label{eq:headgood}	
\sum_{i=1}^L w'_i \leq 1 - C' \kappa \quad \text{and} \quad
\Pr[w'^{(H)} \cdot X^{(H)} \geq \theta - \mu' +
\sqrt{A' \cdot \ln(200/\eps)}\cdot
\newad{\kappa}] \geq
\opt - \eps/40.
\end{equation}

\end{enumerate}

\end{lemma}

\fi

\ifnum\shortsodaversion=1

Let us say that a triple of non-negative
integers $(A,B,C)$ with 
\ifnum\shortsodaversion=1
$A,B,C \leq \poly(1/\kappa)$
\else
$A,B,C \leq M$
\fi
is a \emph{conceivable} triple.
We say that a conceivable triple
$(A,B,C)$ is \emph{achievable} if there exists a
vector $(u_{L+1},\dots,u_n) \in \R^{n-L}_{\geq 0}$ whose coordinates
are non-negative integer multiples of $\kappa$ such that
$\sum_{i=L+1}^n (u_i)^2 = A \kappa^2$,
$\sum_{i=L+1}^n u_i p_i = B \kappa (\eps/(4n)),$ and
$\sum_{i=L+1}^n u_i = C \kappa$, and we
say that such a vector $(u_{L+1},\dots,u_n)$ \emph{achieves} the
triple $(A,B,C).$
Next, our analysis shows that for \emph{any} vector $w''$
with a $\kappa$-granular
tail which matches the $A',B',C'$ values 
\ifnum\shortsodaversion=1
achieved by $(w')^{(T)}$,
\else
from above,
\fi 
the value of its overall
solution is essentially determined by the value that its head
random variable $w''^{(H)} \cdot X^{(H)}$ achieves for the
related-threshold problem.  (See Lemma 13 in the full version.)

\else

Next, our analysis shows that for \emph{any} vector $w''$
with a $\kappa$-granular
tail which matches the $A',B',C'$ values from above, the value of its overall
solution is essentially determined by the value that its head
random variable $w''^{(H)} \cdot X^{(H)}$ achieves for the
related-threshold problem.
More precisely, let us say that a triple of non-negative
integers $(A,B,C)$ with $A,B,C \leq M$ is a \emph{conceivable} triple.
We say that a conceivable triple
$(A,B,C)$ is \emph{achievable} if there exists a
vector $(u_{L+1},\dots,u_n) \in \R^{n-L}_{\geq 0}$ whose coordinates
are non-negative integer multiples of $\kappa$ such that
$\sum_{i=L+1}^n (u_i)^2 = A \kappa^2$,
$\sum_{i=L+1}^n u_i p_i = B \kappa (\eps/(4n)),$ and
$\sum_{i=L+1}^n u_i = C \kappa$, and we
say that such a vector $(u_{L+1},\dots,u_n)$ \emph{achieves} the
triple $(A,B,C).$

\begin{lemma} \label{lem:2}
As above suppose that $v^\ast$ is type $L+1$.
Let $w',A',B',C'$ be as described in Lemma \ref{lem:1}.

Let $w''=(w''_1,\dots,w''_L,w''_{L+1},\dots,w''_n)$ be any vector
with a $\kappa$-granular tail\ignore{whose $L$ head coordinates match
those of $w'$ and} whose $n-L$ tail coordinates $(w''_{L+1},\dots,w''_n)$
achieve the triple $(A',B',C').$ Then
like $w'^{(T)} \cdot X^{(T)}$, the random variable $w''^{(T)} \cdot X^{(T)}$
is strongly concentrated around its mean:
\begin{equation} \label{eq:tailconc2}
\Pr[|w''^{(T)} \cdot X^{(T)} - \mu'| \geq
\newad{\sqrt{A' \cdot \ln(200/\eps)} \cdot \kappa}
] \leq \eps/100,
\end{equation}
and hence
\begin{equation} \label{eq:wprimeprimegood}
\Pr[w'' \cdot X \geq \theta] \geq
\Pr[w''^{(H)} \cdot X^{(H)} \geq \theta - \mu' +\newad{\sqrt{A' \cdot \ln(200/\eps)} \cdot \kappa}
] - \eps/100.
\end{equation}
\end{lemma}
\fi
Intuitively, these two lemmas are useful because they allow us to
``decouple'' the problem of finding an $n$-dimensional solution
vector $w$ into two pieces, finding a head-vector and a tail-vector.
For the tail, these lemmas say that it is enough to search over the
(polynomially many) conceivable triples $(A,B,C)$;
if we can identify the achievable triples from within the conceivable
triples, and for each achievable triple construct \emph{any}
$\kappa$-granular tail vector that achieves it, then this is essentially
as good as finding the actual tail vector of $w'$.  For
the right triple $(A',B',C')$ 
\ifnum\shortsodaversion=1
given by the vector $(w')^{(T)}$ from Lemma 12,
\else
given by \ref{lem:1},
\fi
all that remains is to come up with
a vector of head coordinates that yields a high-value solution
\ifnum\shortsodaversion=1
to the related-threshold problem.
\else
to the related-threshold problem (note that part (3) of Lemma \ref{lem:1}
establishes that indeed such a head-vector must exist).
\fi  
This is
highly reminiscent of Case 1, and indeed we can apply machinery
(the {\tt Find-Optimal-Junta} procedure) from that case for this purpose.
These lemmas thus motivate the two main steps
of the algorithm, Steps 1 and 2, which we describe below.

While there are only polynomially many conceivable triples,
it is a nontrivial task to identify whether any given
conceivable triple is achievable (note that there are exponentially
many different vectors $(u_{L+1},\dots,u_n)$ that might
achieve a given triple).  However, this does turn out to be a
feasible task;
Algorithm {\tt Construct-Achievable-Tails}, called in Step 1
of {\tt Find-Near-Opt-Large-CI}, is
an efficient algorithm (based on dynamic programming) which searches
across all conceivable triples $(A,B,C)$ and identifies those which are
achievable.  For each triple that is found to be achievable,
{\tt Construct-Achievable-Tails}
constructs a $\kappa$-granular tail which achieves it. We have the following lemma:
\begin{lemma} \label{lem:dp}
There is a (deterministic) algorithm {\tt Construct-Achievable-Tails}
that outputs a list consisting precisely of all the achievable $(A,B,C)$
triples, and for each achievable triple it outputs a corresponding tail
vector $(w''_{L+1},\dots,w''_n)$ that achieves it.
The algorithm runs in time $\poly(n,1/\kappa)=\poly(1/\kappa).$
\end{lemma}
Finally, for each achievable triple $(A,B,C)$ and corresponding
tail vector $(w''_{L+1},\dots,w''_n)$ that is generated
by {\tt Construct-Achievable-Tails}, the procedure
{\tt Find-Optimal-Junta} is used to find a setting of the head
coordinates that yields a high-quality solution.

\ignore{
}


\begin{framed}
\noindent {\bf Algorithm
{\tt Find-Near-Opt-Large-CI}:}

\medskip

\noindent {\bf Input:}  probability vector $p=(p_1,\dots,p_n)$ satisfying
(A1) and (A2); parameter $0 < \theta < 1$

\noindent {\bf Output:}  if $v^\ast$ is type $L+1$, a
set ${\cal FEAS}$ of feasible solutions $w$ such that one of them satisfies $\obj(w) \geq
\opt-\eps/2$

\begin{enumerate}

\item Run Algorithm {\tt Construct-Achievable-Tails} to obtain a
list ${\cal T}$ of all achievable triples $(A,B,C)$ and, for each one,
a tail vector $u=(u_{L+1},\dots,u_n)$ that achieves it.

\item For each triple $(A,B,C)$ in ${\cal T}$ and its associated tail
vector $u=(u_{L+1},\dots,u_n)$:
\begin{itemize}
\item  Run {\tt Find-Optimal-Junta}$((p_1,\dots,p_L),\theta - B \kappa \eps/(4n) + \newad{\kappa \cdot \sqrt{\ln(200/\epsilon) \cdot A}},
1-C \kappa)$ to obtain a head $(u_1,\dots,u_L)$.

\item Add the concatenated vector $(u_1,\dots,u_L,u_{L+1},\dots,u_n)$
to the set ${\cal FEAS}$ (initially empty)
of feasible solutions that will be returned.
\end{itemize}
\item Return the set ${\cal FEAS}$ of feasible solutions constructed as
described above.
\end{enumerate}
\end{framed}


\ifnum\shortsodaversion=1
We prove all the aforementioned lemmas, and use them to prove 
Theorem \ref{thm:type-Lplus1},
in the full version.
\else
We prove the aforementioned lemmas in the next subsection.
We conclude this subsection by showing how
Theorem \ref{thm:type-Lplus1} follows from these lemmas.
\medskip
\noindent {\bf Proof of Theorem \ref{thm:type-Lplus1} given
Lemmas \ref{lem:1}, \ref{lem:2}, and \ref{lem:dp}:}
The claimed running time bound is immediate from inspection of
{\tt Find-Near-Opt-Large-CI}, Lemma \ref{lem:dp} (to bound the running
time of {\tt Construct-Achievable-Tails})
and Theorem \ref{thm:L-junta} (to bound the
running time of {\tt Find-Optimal-Junta}).

To prove correctness, suppose that $v^\ast$ is of type $L+1$.  One
of the achievable triples that is listed by {\tt Construct-Achievable-Tails}
will be the $(A',B',C')$ triple that is achieved by the tail
$(w'_{L+1},\dots,w'_n)$ of the vector $w'=(w'_1,\dots,w'_n)$
whose existence is asserted by Lemma \ref{lem:1}.
By Lemma \ref{lem:dp}, {\tt Construct-Achievable-Tails} outputs
this $(A',B',C')$ along with a corresponding tail vector
$(w''_{L+1},\dots,w''_n)$ that achieves it;
by Lemma \ref{lem:2},
any combination $u=(u_1,\dots,u_L,w''_{L+1},\dots,w''_n)$
of a head vector with this tail vector will have $\obj(u) \geq
\Pr[u^{(H)} \cdot X^{(H)} \geq \theta - \mu' + \newad{\kappa \cdot \sqrt{\ln(200/\epsilon) \cdot A'}}
] - \eps/100.$
Lemma \ref{lem:1} ensures that there exists some head vector $w'^{(H)}$
that has $\sum_{i=1}^L w'_i \leq 1 - C' \kappa$ and
$\Pr[w'^{(H)} \cdot X^{(H)} \geq \theta - \mu' + \newad{\kappa \cdot \sqrt{\ln(200/\epsilon) \cdot A'}}] \geq
\opt - \eps/40$, so when {\tt Find-Optimal-Junta}
is called with input parameters $((p_1,\dots,p_L),\theta - B' \kappa
(\eps/(4n)) + \newad{\kappa \cdot \sqrt{\ln(200/\epsilon) \cdot A'}}
, 1 - C' \kappa)$,
by Theorem \ref{thm:L-junta} it
will construct a head $u^{(H)}=(u_1,\dots,u_L)$
with $u_1,\dots,u_L \geq 0$, $u_1 + \cdots + u_L \leq 1-C'\kappa$
which is such that
$\Pr[u^{(H)} \cdot X^{(H)} \geq \theta - \mu' + \newad{\kappa \cdot \sqrt{\ln(200/\epsilon) \cdot A'}}
] \geq
\opt - \eps/40$, and hence the resulting overall vector
$u = (u_1,\dots,u_L,w''_{L+1},\dots,w''_n)$
is a feasible solution which has
$\Pr[u \cdot X \geq \theta ] \geq \opt - 7\eps/200.$
This concludes the proof of Theorem \ref{thm:type-Lplus1}
(modulo the proofs of Lemmas \ref{lem:1}, \ref{lem:2}, and \ref{lem:dp}).
\subsection{Proof of Lemmas
\ref{lem:1}, \ref{lem:2}, and \ref{lem:dp}}
\subsubsection{Proof of Lemma \ref{lem:1}}
Recall from Equation (\ref{eq:L}) that
$
L=L(\eps,\gamma)=\min\{n, \Theta(1/(\eps^2 \gamma^2) \cdot (1/\gamma) \cdot
(\log 1/(\eps\gamma))\cdot(\log(1/\eps))\}$;
since we are in Case 2, we have that
$L=\Theta(1/(\eps^2 \gamma^2) \cdot (1/\gamma) \cdot
(\log 1/(\eps\gamma))\cdot(\log(1/\eps)).$
Since the $\eps\gamma/\new{200}$-critical index of $v^\ast$ is at least $L$,
Lemma 5.5 of \cite{DGJ+:10} gives us that there is a subsequence of weights
$v^{\ast}_{i_1},\dots,v^{\ast}_{i_s}$ with $i_s < L$ and $s \geq t/\gamma$,
where \newad{$t \eqdef \ln(200^2/\eps^3 \gamma)$}, such
that $v^{\ast}_{i_{j+1}} \le v^{\ast}_{i_j}/3$ for all $j=1,\dots,s-1$.
Given this, Claim 5.7 of \cite{DGJ+:10} implies that for any two
points $x \neq x' \in \{0,1\}^s$, we have
\begin{equation} \label{eq:spread}
\left|\sum_{\ell=1}^s v^{\ast}_{i_\ell} x_{i_\ell}- \sum_{\ell=1}^s
v^{\ast}_{i_\ell} x'_{i_\ell}
\right| \ge {\frac {v^{\ast}_{i_s}} 2}.
\end{equation}
(We note that both Lemma 5.5 and Claim 5.7 are simple results with proofs
of a few lines.)
Equation (\ref{eq:spread}) clearly implies that for every $\nu \in \R$
there is at most one $x \in \{0,1\}^s$ such that $\sum_{\ell=1}^s
v^{\ast}_{i_\ell} x_{i_\ell} = \nu$; recalling the definition of $\gamma$,
we further have that
$
\Pr_{(X_{i_1},\dots,X_{i_s}) \sim \mu_{p_{i_1}} \times \cdots \times
\mu_{p_{i_s}}}\left[
\sum_{\ell=1}^s v^{\ast}_{i_\ell} X_{i_\ell} = \nu \right]
\leq (1 - \gamma)^s
$
for every $\nu \in \R$.  Together with (\ref{eq:spread}), this gives
that for every $\nu \in \R$ and every integer $k \geq 0$, we have
\[
\Pr_{(X_{i_1},\dots,X_{i_s}) \sim \mu_{p_{i_1}} \times \cdots \times
\mu_{p_{i_s}}}\left[
\left|\sum_{\ell=1}^s v^{\ast}_{i_\ell} X_{i_\ell} -\nu \right| \le
kv^{\ast}_{i_s}/2 \right]
\leq (2k+1)(1 - \gamma)^s \leq (2k+1)e^{-t} = \newad{(2k+1)\eps^3 \gamma/200^2.}
\]
By independence, using Lemma \ref{lem:extension} we get that this
anti-concentration extends to the linear form over all of the first $L$
coordinates, and hence we get that for all $\nu \in \mathbb{R}$,
\begin{equation} \label{eq:first-L}
\Pr \left[
\left|(v^{\ast})^{(H)} X^{(H)} -\nu \right| \le
kv^{\ast}_{i_s}/2 \right]
\leq (2k+1)\eps^3 \gamma/200^2.
\end{equation}
\ignore{
and indeed to the linear form over all $n$ coordinates,
\begin{equation} \label{eq:all-n}
\Pr_{X \sim \calD_p} \left[
\left|\sum_{\ell=1}^n v^{\ast}_{\ell} X_{\ell} -\theta \right| \le
kv^{\ast}_{i_s}/2 \right]
\leq (2k+1)\eps/200^2.
\end{equation}
}
Now, recall that we are in Case~2 and hence
$\sum_{j > L} v^\ast_j \geq 1/(\new{(L+2)^{(L+2)/2}} + 1).$  Since $v^\ast_{i_s} \geq
v_j$ for all $j > L$, we have that
$v^\ast_{i_s} \geq 1/(n (\new{(L+2)^{(L+2)/2}} + 1))$.
Hence (\ref{eq:first-L}) yields that for all $\nu \in \mathbb{R}$,
\begin{equation} \label{eq:first-L2}
\Pr \left[
\left|(v^{\ast})^{(H)} \cdot X^{(H)} -\nu \right| \le
k/(2n(\new{(L+2)^{(L+2)/2}} + 1)) \right]
\leq \newad{ (2k+1)\eps^3 \gamma/200^2.}
\end{equation}

We now turn from analyzing the head of $v^\ast$ to analyzing the tail.
Recalling again that the $\eps\gamma/\new{200}$-critical
index of $v^\ast$ is greater
than $L$, another application of Lemma 5.5 of \cite{DGJ+:10} gives that
\newad{$\sigma_{L}^2(v^{\ast}) \eqdef \sum_{j > L} (v^\ast_j)^2 \leq
\new{200}^2(v^\ast_{i_s})^2/(\eps^2 \gamma^2)
.$}
The expected value of $(v^{\ast})^{(T)} \cdot X^{(T)}$ is
$\mu = \sum_{j > L} v^{\ast}_j p_j$; an additive Hoeffding bound gives that
for $r>0$,
\[
\Pr[|(v^{\ast})^{(T)} \cdot X^{(T)} - \mu| \geq r \cdot
\newad{\sigma_{L}(v^{\ast})}
]
\leq 2 e^{-r^2}.
\]
Fixing $r = \sqrt{\ln(200/\eps)}$, as a consequence of the above we get that
\[
\Pr[(v^{\ast})^{(T)} \cdot X^{(T)} \geq \mu + \sqrt{\ln(200/\eps)} \cdot
\newad{\sigma_{L}(v^{\ast})}
]
\leq 2 e^{-\ln(200/\eps)} = \eps/100.
\]
Since $\opt = \Pr[v^\ast \cdot X \geq \theta]$, we get that
$$
\Pr[(v^\ast)^{(H)} \cdot X^{(H)} \geq \theta - \mu -
\sqrt{\ln(200/\eps)} \cdot
\newad{\sigma_{L}(v^{\ast})}
] \geq
\opt - \eps/100.
$$
Combining with (\ref{eq:first-L}), we get that
\begin{equation}\label{eq:headalmost}
\Pr[(v^\ast)^{(H)} \cdot X^{(H)} \geq \theta - \mu +
\sqrt{\ln(200/\eps)} \cdot
\newad{\sigma_{L}(v^{\ast})}
] \geq
\opt - \eps/50.
\end{equation}
We are now ready to define the vector $w'$.  Its head coordinates are
the same as $v^\ast$, i.e., for $1 \leq i \leq L$ we have
$w'_i = v^\ast_i$.  We define the quantity
\[
\kappa = 1/(n^2 (\new{(L+2)^{(L+2)/2}}+1)).
\]
For $L+1 \leq i \leq n$, the tail coordinates $w'_i$ of $w'$
are obtained by rounding $v^\ast_i$ down to the nearest integer multiple
of $\kappa$.  It is immediate from this definition that part (1) of the
lemma holds, i.e., $w'$ has a $\kappa$-granular tail and there are non-negative
integers $A,B,C \leq M$ as specified in part (1).
Since $\sum_{i=1}^n w'_i \leq \sum_{i=1}^n v^\ast_i=1$, it must be the
case that $\sum_{i=1}^L w'_i \leq 1 - C' \kappa$, giving the first part of
Equation (\ref{eq:headgood}).

Write $\mu'$ to denote $\E[w'^{(T)} \cdot X^{(T)}] =
\sum_{j>L} w'_j p_j =  B' \kappa (\eps/(4n)).$ Define $\sigma_L^2(w) \eqdef \sum_{j>L} (w'_j)^2$.  \newad{By Hoeffding bound,
we get that  $(w')^{(T)} \cdot X^{(T)}$ is concentrated around its mean $\mu'$. More precisely,
\[
\Pr[|(w')^{(T)} \cdot X^{(T)} - \mu'| \geq
\sqrt{\ln(200/\eps)}  \cdot
\sigma_L(w)
]
\leq 2 e^{-\ln(200/\eps)} \leq \eps/100,
\]
giving part (2) of Lemma \ref{lem:1}.  Note that $\sigma_L^2(w) \le \sigma_L^2(v^{\ast}) \le \new{200}^2(v^\ast_{i_s})^2/(\eps^2 \gamma^2)$.  }

It remains only to establish the second part of
Equation (\ref{eq:headgood}).  Equation (\ref{eq:headalmost}) almost
gives us this -- it falls short only in having $\mu$ in place of
$\mu'$ in the lower bound for
$(w')^{(H)} \cdot X^{(H)}$ (recall that $(v^\ast)^{(H)}$ is identical
to $(w')^{(H)}$).  To get around this we use the anti-concentration
property of the head that was established in (\ref{eq:first-L2}) above.
Since $|\mu - \mu'| \leq n \kappa = 1/(n(\new{(L+2)^{(L+2)/2}}+1))$, equation
(\ref{eq:first-L2}) gives that
\[
\Pr[(w')^{(H)} \cdot X^{(H)} \in
[
\theta - \mu +
\sqrt{\ln(200/\eps)} \cdot
\newad{\sigma_L(w)} ,
\theta - \mu' +
\sqrt{\ln(200/\eps)} \cdot
\newad{\sigma_L(w)}
]
]
\leq \eps/200
\]
and combining this with (\ref{eq:headalmost}) gives
\[
\Pr[(w')^{(H)} \cdot X^{(H)} \geq \theta - \mu' +
\sqrt{\ln(200/\eps)} \cdot
\newad{\sigma_L(w)} ] \geq
\opt - 5\eps/200,
\]
the desired second part of
Equation (\ref{eq:headgood}).
This concludes the proof of Lemma \ref{lem:1}.

\subsubsection{Proof of Lemma \ref{lem:2}}
Since by assumption the tail of $w''$ achieves the triple $(A',B',C')$, we
have that the mean $\E[(w'')^{(T)} \cdot X^{(T)}]$ equals
$B'\kappa(\eps/(4n))$ and thus is the same as $\mu'$, the mean of
$(w')^{(T)} \cdot X^{(T)}]$.  Since $\sum_{j>L} (w''_j)^2=
\sum_{j>L}(w'_j)^2$, just as was the case for $w'$ we get that a Hoeffding
bound gives the desired concentration bound,
\[
\Pr[|w''^{(T)} \cdot X^{(T)} - \mu'| \geq \newad{\kappa \cdot \sqrt{\ln(200/\epsilon) \cdot A'}}
] \leq \eps/100.
\]
Thus, we have established Equation (\ref{eq:tailconc2}).

Equation (\ref{eq:tailconc2}) implies that $w''^{(T)} \cdot X^{(T)} <
\mu' -
\newad{\kappa \cdot \sqrt{\ln(200/\epsilon) \cdot A'}}
$ with probability at most $\eps/100.$  Since
$w''^{(H)} \cdot X^{(H)} \geq \theta - \mu' +
\newad{\kappa \cdot \sqrt{\ln(200/\epsilon) \cdot A'}}
$ and
$w''^{(T)} \cdot X^{(T)} \geq \mu' -
\newad{\kappa \cdot \sqrt{\ln(200/\epsilon) \cdot A'}}
$ together imply that
$w'' \cdot X \geq \theta$, we thus get
Equation (\ref{eq:wprimeprimegood}), and the lemma is proved.

\subsubsection{Proof of Lemma \ref{lem:dp}}

The {\tt Construct-Achievable-Tails} algorithm is based on dynamic programming.
Let $w=(w_{L+1},\dots,w_n)$ be a tail weight vector such that each
$w_i$ is a non-negative integer multiple of $\kappa$.
We define the quantities

\[
A(w) = \sum_{i>L} (w_i)^2/\kappa^2; \quad
B(w) = \sum_{i>L} w_i p_i / (\kappa \eps/(4n)); \quad
C(w) = \sum_{i>L} w_i/\kappa.
\]

Recalling Assumption (A2), we see that each of $A(w),B(w),
C(w)$ is a non-negative integer \newad{bounded by $\poly(1/\kappa)$}.

For each conceivable triple $(A,B,C)$ and for every $t \in \{L+1,\dots,n\}$,
we create a sub-problem in which the goal is to determine whether there is
a choice of weights $w_{L+1},\dots, w_t$
(each of which is a non-negative integer multiple of $\kappa$, with
all other weights $w_{t+1},\dots,w_n$ set to 0) such that $A(w)=A,B(w)=B,$
and $C(w)=C$.
Such a choice of weights $w_{L+1},\dots,w_t$ exists if and only if
there is a nonnegative-integer-multiple-of-$\kappa$ choice of $w_t$ for which
there is a nonnegative-integer-multiple-of-$\kappa$
choice of weights $w_{L+1},\dots,w_{t-1}$ (with all subsequent weights set to
0) such that $A(w)=A - (w_t)^2/\kappa$, $B(w)=B-w_tp_t/(\kappa \eps/(4n))$,
and $C(w)=C - w_t/\kappa.$

Thus, given the set of all triples that are achievable with only weights
$w_{L+1}, \dots,w_{t-1}$ allowed to be nonzero, it is straightforward to
efficiently
(in $\poly(1/\kappa)$ time)
identify the set of all triples that are achievable with only
weights $w_{L+1},\dots,w_t$ allowed to be nonzero.  This is because for a given
candidate (conceivable) triple $(A,B,C)$, one
can check over all possible values of $w_{t}$ (that are integer multiples
of $\kappa$ and upper bounded by 1) whether the triple
$(A - (w_t)^2/\kappa,B-w_tp_t/(\kappa \eps/(4n)),C - w_t/\kappa)$ is
achievable with only weights
$w_{L+1}, \dots,w_{t-1}$ allowed to be nonzero.  Since there are
only $O(1/\kappa)$ choices of the weight $w_{t}$ and the overall
number of sub-problems in this dynamic program is bounded by $\poly(n,1/\kappa)
= \poly(1/\kappa)$, the overall entire dynamic program runs
in $\poly(1/\kappa)$ time.  This concludes the proof of Lemma \ref{lem:dp}.

\fi


\section{Case 3:  $v^\ast$ is type $K$ for some $1 \leq K \leq L$}
\label{sec:type-i}
\ignore{
%
}
Recall that in Case 3 the optimal solution $v^\ast$ is not an $L$-junta,
so it satisfies
$\littlesum_{i=1}^L v^\ast_i \leq \new{(L+2)^{(L+2)/2}} \cdot \littlesum_{i=L+1}^n v^\ast_i $, 
and $c(v^\ast,\eps) = K$ for some $1 \leq K \leq L.$
For this case we prove the following theorem:
\begin{theorem} \label{thm:type-K}
There is a randomized algorithm {\tt Find-Near-Opt-Small-CI} with the
following performance guarantee:
The algorithm takes as input a vector of probabilities
$p=(p_1,\dots,p_n)$ satisfying
(A1) and (A2), a threshold value $0 < \theta < 1$, a value
$1 \leq K \leq L$, and a confidence parameter $0 < \delta < 1.$
It runs in $\new{\poly(n,2^{\poly(L)},\bit(\theta)) \cdot \log(1/\delta)}$ time
and outputs a set of $N \leq \poly(n,2^{\poly(L)})$ many feasible
solutions.
If $v^\ast$ is of type $K$ then with probability $1-\delta$
one of the feasible solutions $w$
that it outputs satisfies $\obj(w) \geq \opt - \eps/2.$
\end{theorem}

\ifnum\shortsodaversion=1
This case uses a range of tools from probability (the 
Dvoretsky-Kiefer-Wolfowitz inequality, the Berry-Ess{\'e}en theorem, and Gaussian
anti-concentration) together with extensions of the {\tt Construct-Achievable-Tails}
and {\tt Find-Optimal-Junta} procedures described in previous sections.  Because of 
space constraints we omit it here (see the full version).

\else

\subsection{Useful probabilistic tools and notation.}

\paragraph{Kolmogorov distance.}  For $X,Y$ two real-valued random variables
we say the \emph{Kolmogorov distance} $\dK(X,Y)$ between $X$ and $Y$
is $\dK(X,Y) \eqdef \sup_{t \in \R}|\Pr[X \leq t] - \Pr[Y \leq t]|.$
\smallskip

\noindent{\bf Remark.} If $w$ is an optimal solution \new{of problem (P)} and the random variables
$w \cdot X, w' \cdot X$ have Kolmogorov distance at most $\eps$
then $\obj(w') \geq \opt - \eps.$

We recall the following useful elementary fact about Kolmogorov distance:
\begin{fact} \label{fact:dK-indep}
Let $X,Y,Z$ be real-valued random variables such that $X$ is independent
of $Y$ and independent of $Z$.  Then we have that $\dK(X+Y,X+Z) \leq
\dK(Y,Z).$
\end{fact}
The \emph{Dvoretsky-Kiefer-Wolfowitz (DKW) inequality} is a
considerably more sophisticated fact about Kolmogorov distance that
will also be useful.
Given $m$ independent samples $t_1,\dots,t_m$ drawn from a real-valued
random variable $X$,
the {\em empirical distribution} $\wh{X}_m$ is defined as
the real-valued random variable which is uniform over the multiset
$\{t_1,\dots,t_m\}.$
The DKW inequality states that for $m=\Omega((1/\eps^2)\cdot \ln(1/\delta))$,
with probability $1-\delta$ the {empirical distribution} $\wh{X}_m$ will be
$\eps$-close to $p$ in Kolmogorov distance:
\begin{theorem}[\cite{DKW56,Massart90}] \label{thm:DKW}
For all $\eps>0$ and any real-valued random variable $X$, we have
$\Pr [ \dK(p, \wh{p}_m) > \eps ] \leq 2e^{-2m\eps^2}.$
\end{theorem}
\noindent We will also require a corollary of the Berry-Ess{\'e}en
theorem (see e.g.,~\cite{Feller}).  We begin by recalling the theorem:
\begin{theorem} \label{thm:be} (Berry-Ess{\'e}en)  Let
$X_1, \dots, X_n$ be a sequence of independent random variables
satisfying $\E[X_i] = 0$ for all $i$, ${\sum_i \E[X_i^2]} =
\sigma^2$, and $\sum_i \E[|X_i|^3] = \rho_3$.  Let $S = X_1 + \cdots
+ X_n$ and let $F$ denote the cumulative distribution
function (cdf) of $S$. Then
\[
\sup_x |F(x) - \Phi_\sigma(x)| \leq C\rho_3/\sigma^3,
\]
where $\Phi_\sigma$ is the cdf of a $N(0,\sigma^2)$
Gaussian random variable (with mean zero and variance $\sigma^2$),
and $C$ is a universal constant.
\cite{Shiganov:86} has shown that one can take $C = .7915$.
\end{theorem}

\begin{corollary}\label{cor:BE}
Let $X=(X_1, \dots, X_n) \sim \calD_p$ and suppose that
$\min_{i \in [n]}\{p_i,1-p_i\} \geq \gamma > 0.$
Let $w \in \R^n$ be $\tau$-regular.
Let $Z$ be the random variable $w \cdot X$ and define
$\mu = \E[w \cdot X] = \sum_{i=1}^n w_i p_i$,
$\sigma^2 = \Var[w \cdot X] = \sum_{i=1}^n w_i^2 \cdot p_i(1-p_i)$. Then
$\dK(Z,N(\mu,\sigma^2)) \leq \eta$ where $\eta = \tau/\gamma$.
\end{corollary}
\begin{proof}
Define the random variable $Y_i = w_i(X_i - p_i)$, so $\E[Y_i]=0.$
It suffices to show that the random variable
$Y = \sum_{i=1}^n w_i Y_i$ has $\dK(Y,N(0,\sigma^2)).$
We have
$\sum_i \E[Y_i^2] = \sigma^2 = \sum_{i=1}^n w_i^2 p_i(1-p_i)$
and
$$
\mathbf{E} [|y_i|^3] =   w_i^3 \left( p_i \cdot (1-p_i)^3 + (1-p_i)
\cdot (p_i)^3 \right) = w_i^3 p_i(1-p_i) \cdot (p_i^2 + (1-p_i)^2), \text{~so}
$$
$$
\sum_{i=1}^n \E[|Y_i|^3] = \sum_{i=1}^n w_i^3 p_i(1-p_i)
(p_i^2 + (1-p_i)^2)
\leq \sum_{i=1}^n w_i^3 p_i(1-p_i).
$$
The Berry-Ess{\'e}en theorem thus gives
\begin{eqnarray*}
\dK(Y,N(0,\sigma^2)) &\leq&
{\frac { \sum_{i=1}^n w_i^3 p_i(1-p_i)} {\left(
\sum_{i=1}^n w_i^2 p_i(1-p_i)\right)^{3/2}}}
\leq \max_{i=1}^n |w_i| \cdot
{\frac { \sum_{i=1}^n w_i^2 p_i(1-p_i)} {\left(
\sum_{i=1}^n w_i^2 p_i(1-p_i)\right)^{3/2}}}=
\max_{i=1}^n |w_i| \cdot {\frac 1 \sigma}.
\end{eqnarray*}
Recalling that (by regularity) we have $\max_i w_i \leq \tau \sqrt{\sum_i
w_i^2}$, and that by definition of $\gamma$ and $\sigma$ we have $\sigma \geq
\gamma \sqrt{\sum_i w_i^2}$, we get that
$\max_{i=1}^n |w_i| \cdot {\frac 1 \sigma} \leq \tau/\gamma$ as desired.
\end{proof}

Finally, we recall the well-known fact that an $N(\mu,\sigma^2)$ Gaussian
is $\eps$-anti-concentrated at radius $\eps \sigma$ (this follows
directly from the fact that the pdf of an $N(\mu,\sigma^2)$ Gaussian
is given by ${\frac 1 {\sigma \sqrt{2 \pi}}} \exp\left(
-{\frac {(x-\mu)^2}{2 \sigma^2}}\right)$).

\medskip

\noindent {\bf Notation.}
In this section our analysis will deal separately with the coordinates
$1,\dots,K$ and the coordinates $K+1,\dots,n$, so we use the following notational conventions.
For an $n$-dimensional vector $w \in \R^n$, in this section
we refer to $(w_1,\dots,w_{K-1})$ as the ``head''
of $w$ and we write $w^{(H)}$ to denote this vector; similarly
we write $w^{(T)}$ to denote the ``tail'' $(w_{K},\dots,w_n)$ of $w$.
We sometimes refer to a vector in $\R^{K-1}$ as a ``head vector'' and to a
vector in $\R^{n-K+1}$ as a ``tail vector.''
In a random variable $w^{(H)} \cdot X^{(H)}$ the randomness is over the
draw of $X^{(H)} \sim \bigotimes_{i=1}^{K-1} \mu_{p_i}$,
and similarly for the random variable $w^{(T)} \cdot X^{(T)}$
the randomness is over the draw of $X^{(T)}
\sim \bigotimes_{i=K}^n \mu_{p_i}.$

We additionally modify some of the terminology from Section
\ref{sec:type-Lplus1} dealing with granular vectors and achievable
triples.  
\ifnum\shortsodaversion=1
Fix $\kappa = \poly(1/n,1/2^{\poly(L)})$.
\else
Fix $\kappa = \poly(1/n,1/2^{\poly(L)})$
(we give a more precise value of $\kappa$ later).
\fi
We say that a vector $w=(w_1,\dots,w_n)
\in \R^n$ has a \emph{$\kappa$-granular tail} if each coordinate
$w_i$, $K \leq i \leq n$, is an integer multiple of $\kappa$.
It is easy to see that for any vector $w \in \R^{n}_{\leq 0}$
with $\sum_{i=1}^n w_i \leq 1$ that has a $\kappa$-granular
tail, for $M \eqdef \poly(1/\kappa)$
there must exist non-negative integers $A,B,C \leq M$ such that
$\E[w^{(T)} \cdot X^{(T)}] =
\sum_{i=K}^n w_i p_i = A \kappa(\eps/(4n))$,
$\Var[w^{(T)} \cdot X^{(T)}] = \sum_{i=K}^n w_i^2 p_i (1-p_i) = B \kappa^2
(\eps/(4n))^2$, and
$\sum_{i=K}^n w_i = C' \kappa.$
We say that a triple of non-negative
integers $(A,B,C)$ with $A,B,C \leq M$ is a \emph{conceivable} triple.
We say that a conceivable triple
$(A,B,C)$ is \emph{$\eps'$-regular achievable} if there exists an
$\eps'$-regular
vector $u^{(T)}=(u_{K+1},\dots,u_n) \in \R^{n-K+1}_{\geq 0}$ whose coordinates
are non-negative integer multiples of $\kappa$ such that
$\E[u^{(T)} \cdot X^{(T)}]  = A \kappa(\eps/(4n))$,
$\Var[u^{(T)} \cdot X^{(T)}] = B \kappa^2 (\eps/(4n))^2$, and
$\sum_{i=K}^n u_i = C \kappa,$
and we
say that such a vector $(u_{L+1},\dots,u_n)$ \emph{achieves} the
triple $(A,B,C).$

\subsection{The algorithm and an intuitive explanation of its performance.}


\begin{framed}
\noindent {\bf Algorithm
{\tt Find-Near-Opt-Small-CI}:}

\medskip

\noindent {\bf Input:}  probability vector $p=(p_1,\dots,p_n)$ satisfying
(A1) and (A2); parameter $0 < \theta < 1$; parameter $1 \leq K \leq L$;
confidence parameter $0 < \delta < 1$

\noindent {\bf Output:}  if $v^\ast$ is type $K$,
a set ${\cal FEAS}$ of feasible solutions $w$ such that
one of them satisfies $\obj(w) \geq \opt-\eps/2$

\begin{enumerate}

\item Run Algorithm {\tt Construct-Achievable-Regular-Tails}$(\eps\gamma
/\new{100})$
to obtain a list ${\cal T}$ of all triples $(A,B,C)$ that are achieved
by some $\eps\gamma/\new{100}$-regular tail vector and, and, for each one,
an $\eps\gamma/\new{100}$-regular tail vector $u=(u_{L+1},\dots,u_n)$ that achieves it.

\item For each triple $(A,B,C)$ in ${\cal T}$ and
its associated tail vector $u=(u_{K},\dots,u_n)$,

\begin{itemize}

\item Run {\tt Find-Approximately-Best-Head}$(u_K,\dots,u_n,\eps/200,\delta/(2|{\cal T}|))$ to obtain a head
vector $(u_1,\dots,u_{K-1})$

\item Add the concatenated vector $(u_1,\dots,u_{K-1},u_K,\dots,u_n)$
to the set ${\cal FEAS}$ (initially empty) of feasible solutions
that will be returned.

\end{itemize}

\item Return the set ${\cal FEAS}$
of feasible solutions constructed as described
above.

\end{enumerate}

\end{framed}


Similar to Case 2, the high level idea of this case is to decouple
the problem of finding a good solution into two pieces, namely
finding a good tail and finding a good head.  However,
in Case 2 the anti-concentration of the head random variable
(see Equation  (\ref{eq:first-L2})) played an essential role;
in contrast, here in Case 3 the fact that the tail random variable is close to a
Gaussian will play the key role.  At a high level,
the analysis for this case proceeds as follows.

First, using the facts that the vector $(v^\ast_{K},\dots,v^\ast_n)$
is $\eps\gamma/\new{200}$-regular and that
$\littlesum_{i=1}^L v^\ast_i \leq \new{(L+2)^{(L+2)/2}} \cdot \littlesum_{i=L+1}^n v^\ast_i$,
 we get that the tail random variable
$(v^\ast)^{(T)} \cdot X^{(T)}$
is $O(\eps)$-close to a Gaussian $N(\mu,\sigma^2)$
in Kolmogorov distance, where the variance $\sigma^2$ is ``not too small''
(see Lemma \ref{lem:closetogaussian}).
Next, we argue that for any head vector $(w')^{(H)}=
(w'_1,\dots,w'_{K-1})$, there exists a tail vector
$(w')^{(T)}=(w'_K,\dots,w'_n)$, obtained by rounding the
tail coordinates $v^\ast_K,\dots,v^\ast_n$ down to some not-too-small
granularity $\kappa$, which is ``nice'' (i.e., regular and with
not-too-small variance) and which
gives a solution of
almost equal quality to what would be obtained by having the
actual $(v^\ast_K,\dots,v^\ast_n)$ as the tail weights
(see Lemma \ref{lem:closetogranular}).
We then strengthen this by showing that for any head vector,
\emph{any} tail vector which is regular and has the right mean and
variance similarly gives a solution of
almost equal quality to what would be obtained by having the
actual $(v^\ast_K,\dots,v^\ast_n)$ as the tail weights
(see Lemma \ref{lem:anytail}).
This motivates the
{\tt Construct-Achievable-Regular-Tails} procedure (called in Step 1); it
uses dynamic programming to efficiently search across all conceivable triples
and identify precisely those that are achieved by some
$\eps\gamma/100$-regular $\kappa$-granular tail vector (and for each achievable
triple, identify a tail vector $(u_{K},\dots,u_n)$ that achieves it).

Intuitively, at this point the algorithm has identified a polynomial-sized
collection of tail vectors one of which ``is good'' (does almost
as well as the optimal tail vector $(v^\ast_K,\dots,v^\ast_n)$
if it were paired with the optimal head vector).  It remains to show
that it is possible to find a high-quality head vector and that
combining such a head vector with this ``good''
tail vector yields an overall high-quality solution.
We do this, and conclude the proof of Theorem \ref{thm:type-K},
in Section ~\ref{sec:case3-head}.

\subsection{Good tails exist and can be found efficiently:
Proofs of Lemmas \ref{lem:closetogaussian} -- \ref{lem:anytail}
and analysis of {\tt Construct-Achievable-Regular-Tails}}

Let
\begin{equation} \label{eq:mudef}
\mu \eqdef \E[
(v^\ast)^{(T)} \cdot X^{(T)}] = \sum_{i=K}^n v^\ast_i p_i
\quad \text{and} \quad
\sigma^2 \eqdef \Var[ (v^\ast)^{(T)} \cdot X^{(T)}] =
\sum_{i=K}^n (v^\ast_i)^2 p_i(1-p_i)
.
\end{equation}

\begin{lemma} \label{lem:closetogaussian}
Suppose $v^\ast$ is type $K$.  Then $\dK(
(v^\ast)^{(T)} \cdot X^{(T)},N(\mu,\sigma^2)) \leq \eps/\new{200}$, and
$\sigma \geq {\frac \gamma {(\new{(L+2)^{(L+2)/2}} + 1) n}}.$
\end{lemma}
\begin{proof}
Since $v^\ast$ is type $K$, we have that $(v^\ast)^{(T)}$ is
$\new{\eps\gamma/200}$-regular, and hence
Corollary \ref{cor:BE} gives that $\dK
((v^\ast)^{(T)} \cdot X^{(T)},N(\mu,\sigma^2)) \leq \eps/200.$

For the lower bound on $\sigma$, we observe that since $K \leq L$,
$\littlesum_{i=1}^L v^\ast_i \leq \new{(L+2)^{(L+2)/2}}  \littlesum_{i=L+1}^n v^\ast_{i}$, and $\littlesum_{i=1}^n v^\ast_i=1$, we have
\[
v^\ast_K + \cdots + v^\ast_n \geq v^\ast_{L+1} + \cdots + v^\ast_n
\geq {\frac 1 {(\new{(L+2)^{(L+2)/2}}  + 1)}}.
\]
Hence Cauchy-Schwarz implies that
\[
\sqrt{\sum_{i=K}^n (v^\ast_i)^2}
\geq {\frac 1 {(\new{(L+2)^{(L+2)/2}}  + 1) (n-K)}}
\geq {\frac 1 {(\new{(L+2)^{(L+2)/2}}  + 1) n}}\]  
so
\[\sigma = \sqrt{\sum_{i=K}^n (v^\ast_i)^2 p_i(1-p_i)} \geq
{\frac \gamma {(\new{(L+2)^{(L+2)/2}}  + 1) n}}.
\]
\end{proof}

We now define the value of $\kappa$ to be
\[
\kappa = \new{{\frac {\eps \gamma^2}{200(\new{(L+2)^{(L+2)/2}} +1)^2 n^3}}}.
\]

\begin{lemma} \label{lem:closetogranular}
As above suppose $v^\ast$ is type $K$.  Let $w' \in\R^n_{\geq 0}$ be
a feasible solution which is such that for $K \leq i \leq n$, the value
$w'_i$ is obtained from $v^\ast_i$ by rounding down to the nearest
integer multiple of $\kappa$.
Then
\begin{enumerate}
\item The vector $(w')^{(T)}=(w'_K,\dots,w'_n)$ is
$\eps\gamma/\new{100}$-regular;

\item The variance $(\sigma')^2 \eqdef
\Var[(w')^{(T)} \cdot X^{(T)}]$ is at least
${\frac 1 2} \sigma^2 \geq
{\frac 1 2} \cdot {\frac {\gamma^2} {(\new{(L+2)^{(L+2)/2}} +1)^2n^2}}$; and

\item $\obj(w') \geq \obj(w'_1,\dots,w'_{K-1},v^\ast_K,
\dots,v^\ast_n) - \eps/40.$

\end{enumerate}

\end{lemma}

\begin{proof}
We start by lower bounding $(\sigma')^2$ as follows.  Since\ignore{$\sigma^2
= \sum_{i=K}^n (v^\ast_i)^2 p_i(1-p_i) \geq
{\frac {\gamma^2}{(L!+1)^2n^2}}$
and}
each $w'_i$, $K \leq i \leq n$, is less than $v^\ast_i$ by at most $\kappa$,
we have that $\sum_{i=K}^n (w'_i)^2$ is less than
$\sum_{i=K}^n (v^\ast_i)^2$ by at most $2 \kappa n$
and hence
\[
\sigma^2 - (\sigma')^2 \leq 2 \kappa n \cdot \max_{i=K}^n p_i(1-p_i)
\leq  \new{{\frac {\kappa n} 2} <
{\frac 1 2} \cdot {\frac {\gamma^2}{(\new{(L+2)^{(L+2)/2}} +1)^2 n^2}}
\leq {\frac 1 2} \cdot \sigma^2}
\quad \text{so}\quad
(\sigma')^2 \geq {\frac 1 2} \sigma^2,
\]
giving (2).
Part (1) follows easily from (2) and the fact that
$w'_i \leq v^\ast_i$ for $K \leq i \leq n.$

For part (3) we use the fact that the tail $w'^{(T)} \cdot X^{(T)}$
is anti-concentrated (since, by regularity, it is close to a Gaussian).
In more detail, fix an outcome $(y_1,\dots,y_{K-1}) \in \{0,1\}^{K-1}$
for the head bits.  Since $\sum_{i=K}^n w'_i y_i \geq
\sum_{i=K}^n v^\ast_i y_i - \kappa n$ for all $(y_K,\dots,y_n)
\in \{0,1\}^{n-k+1},$
we have
\begin{eqnarray}
&&
\Pr\left[\sum_{j=1}^{K-1} w'_j y_j + (v^\ast_i)^{(T)} \cdot X^{(T)} \geq \theta \right]
-
\Pr\left[ \sum_{j=1}^{K-1} w'_j y_j + (w')^{(T)} \cdot X^{(T)} \geq \theta\right]
\nonumber\\
&\leq&
\Pr \left[ (w')^{(T)} \cdot X^{(T)} \in \Big[ \theta - \sum_{j=1}^{K-1} w'_j y_j -
 \kappa n, \theta - \sum_{j=1}^{K-1} w'_j y_j \Big] \right].
\label{eq:z}
\end{eqnarray}
Since by (1) we know that $(w')^{(T)}$ is $\eps \gamma/\new{100}$-regular,
Corollary \ref{cor:BE} gives us that $$\dK \left( (w')^{(T)} \cdot X^{(T)},
N\big( \E[(w')^{(T)} \cdot X^{(T)}], (\sigma')^2 \big) \right) \leq \eps/100.$$
Since $\kappa n / 2 \leq \eps \sigma'/\new{200}$,
as noted after Lemma \ref{lem:extension}
a random variable $Z \sim
N(\E[(w')^{(T)} \cdot X^{(T)}],(\sigma')^2)$ has
$\Pr[Z \in I] \leq \eps/200$ for any interval $I$ of length $\kappa n.$
Hence (\ref{eq:z}) is at most ${\frac \eps {100}} + {\frac \eps {100}}
+ {\frac \eps {200}} = {\frac \eps {40}}$.  Since this holds for
each fixed $(y_1,\dots,y_{K-1}) \in \{0,1\}^{K-1}$, we get (3).
\end{proof}

\begin{lemma} \label{lem:anytail}
As above suppose $v^\ast$ is type $K$.  Fix $(w'')^{(T)}
= (w''_K,\dots,w''_n)
\in \R^{n-K+1}_{\geq 0}$ to be \emph{any}
$\eps\gamma/\new{100}$-regular tail vector such that
$\mu'' \eqdef \E[(w'')^{(T)} \cdot X^{(T)}]$ equals
$\mu' \eqdef \E[(w')^{(T)} \cdot X^{(T)}]$,
and $(\sigma'')^2 \eqdef \Var[(w'')^{(T)} \cdot X^{(T)}]$
equals $(\sigma')^2$ (see part (2) of
\newad{Lemma \ref{lem:closetogranular}).}
Then for any head vector $(w'')^{(H)}=(w''_1,\dots,w''_{K-1})$,
we have that $\obj((w''_1,\dots,w''_{K-1},w''_K,\dots,w''_n)) \geq
\obj(w''_1,\dots,w''_{K-1},v^\ast_K,\dots,v^\ast_n)- \eps/40.$
\end{lemma}
\begin{proof} The proof is identical to part (3) of Lemma
\ref{lem:closetogranular}.
\end{proof}

Having established the existence of a ``good'' tail (the vector $(w')^{(T)}$
from Lemma \ref{lem:closetogranular}), we now argue that
{\tt Construct-Achievable-Regular-Tails}
can efficiently construct a list containing some such good tail vector.
Lemma \ref{lem:anytail} ensures that finding any such good tail vector
is as good as finding the actual tail vector $(w')^{(T)}$
obtained from $(v^\ast)^{(T)}$ by rounding down as
described in Lemma \ref{lem:closetogranular}.

\begin{lemma} \label{lem:CART}
There is a (deterministic) algorithm
{\tt Construct-Achievable-Regular-Tails}($\eps')$
that, given input parameters $\eps'$ \newad{and $K$}, outputs a
list consisting precisely of all the $\eps'$-regular achievable $(A,B,C)$
triples, and for each achievable triple it outputs a corresponding tail
vector $(w''_{K},\dots,w''_n)$ that achieves it.
The algorithm runs in time $\poly(n,1/\kappa)=\poly(1/\kappa).$
\end{lemma}

\begin{proof}
Similar to the earlier {\tt Construct-Achievable-Tails} algorithm, the
main idea is to use dynamic programming; however the details are
somewhat different, chiefly because of the need to ensure
regularity (and also because the numerical quantities involved
are somewhat different from before).

Let \newad{$w=(w_{K},\dots,w_n)$} be a tail weight vector such that each
$w_i$ is a non-negative integer multiple of $\kappa$.
We define the quantities
\begin{eqnarray*}
&&
A(w) = \sum_{i=K}^n w_i p_i / (\kappa \eps/(4n)); \quad
B(w) = \sum_{i=K}^n w_i^2 p_i(1-p_i) / (\kappa^2 (\eps/(4n))^2); \quad
C(w) = \sum_{i=K}^n w_i/\kappa; \\
&&
D(w) = \sum_{i=K}^n w_i^2 / \kappa^2; \quad
E(w) = \max_{i=K}^n w_i/\kappa.
\end{eqnarray*}
Recalling Assumption (A2), we see that each of $A(w),B(w),
C(w),D(w),E(w)$ is a non-negative integer.
We say that a quintuple $(A,B,C,D,E)$ is \emph{conceivable} if all values are
non-negative integers at most $M$.

For each conceivable quintuple $(A,B,C,D,E)$
and for every $t \in \{K,\dots,n\}$,
we create a sub-problem in which the goal is to determine whether there is
a choice of weights $w_K,\dots, w_t$
(each of which is a non-negative integer multiple of $\kappa$, with
all other weights $w_{t+1},\dots,w_n$ set to 0) such that $A(w)=A,$
$B(w)=B,$ $C(w)=C,$ $D(w)=D$ and $E(w)=E.$
Such a choice of weights $w_{K},\dots,w_t$ exists if and only if
there is a nonnegative-integer-multiple-of-$\kappa$ choice of $w_t$ for which
there is a nonnegative-integer-multiple-of-$\kappa$
choice of weights $w_{K},\dots,w_{t-1}$ (with all subsequent weights set to
0) such that $A(w)=A -w_tp_t/(\kappa \eps/(4n))$, $B(w)=B-w_t^2p_t(1-p_t)/
(\kappa^2 (\eps/(4n))^2)$, $C(w)=C - w_t/\kappa$,
$D(w) = D - w_t^2/\kappa^2$, and $E=\max\{E(w),w_t/\kappa\}.$

Thus, given the set of all quintuples that are achievable with only weights
$w_{K}, \dots,w_{t-1}$ allowed to be nonzero, it is straightforward to
efficiently
(in $\poly(1/\kappa)$ time)
identify the set of all quintuples that are achievable with only
weights $w_{K},\dots,w_t$ allowed to be nonzero.
\ignore{ This is because for a given
candidate (conceivable) quintuple $(A,B,C,D,E)$, one
can check over all possible values of $w_{t}$ (that are integer multiples
of $\kappa$ and upper bounded by 1) whether the triple
$(A - (w_t)^2/\kappa,B-w_tp_t/(\kappa \eps/(4n)),C - w_t/\kappa)$ is
achievable with only weights
$w_{L+1}, \dots,w_{t-1}$ allowed to be nonzero.}Since there are
only $O(1/\kappa)$ choices of the weight $w_{t}$ and the overall
number of sub-problems in this dynamic program is bounded by $\poly(n,1/\kappa)
= \poly(1/\kappa)$, the overall entire dynamic program runs
in $\poly(1/\kappa)$ time.

Once the set of all achievable quintuples has been obtained, it is
straightforward for each quintuple $(A,B,C,D,E)$ to determine whether
or not it is $\eps'$-regular (by computing $E/\sqrt{D}$ and comparing
against $\eps'$).  Having identified the set of all $\eps'$-regular
quintuples, it is \newad{easy} to output a list consisting of
all the $\eps'$-regular achievable $(A,B,C)$ triples (and from the
dynamic program it is easy to maintain a tail vector achieving the
triple in the usual way).  This concludes the proof of Lemma \ref{lem:CART}.
\end{proof}

\ignore{

\begin{lemma} \label{lem:anyhead}
As above suppose $v^\ast$ is type $K$.
Fix any head vector $(w'')^{(H)}=(w''_1,\dots,w''_{K-1})$ such that
$\sum_{i=1}^{K-1} |w''_i - v^\ast_i| \leq \new{blah}$ and
fix the tail vector $(w'')^{(T)}=(w''_K,\dots,w''_n)$ to be as
in Lemma \ref{lem:anytail}.
Then $w''=(w''_1,\dots,w''_n)$ has $\obj(w'') \geq \opt - \eps/20.$
\end{lemma}

}

\subsection{
Finding a good head vector:  The {\tt Find-Approximately-Best-Head} procedure
and the proof of Theorem \ref{thm:type-K}}
\label{sec:case3-head}

By Lemma \ref{lem:CART}
the {\tt Construct-Achievable-Regular-Tails} procedure generates a tail
vector $(w'')^{(T)}$ that matches the mean, variance and $L_1$-norm
of the $(w')^{(T)}$ vector whose existence is asserted by
Lemma \ref{lem:closetogranular}.
In the rest of this section we consider the execution of {\tt
Find-Approximately-Best-Head} when it is run on this tail
vector $(w'')^{(T)}$ as input.

By the DKW inequality (Theorem \ref{thm:DKW}), with high probability the random variable $R$
has $\dK(R,(w'')^{(T)} \cdot X^{(T)}) \leq \eps/200$; we henceforth
assume that this is indeed the case.
Fact \ref{fact:dK-indep} implies that
$\dK((v^\ast)^{(H)} \cdot X^{(H)} + R,
(v^\ast)^{(H)} \cdot X^{(H)} + (w'')^{(T)} \cdot X^{(T)}) \leq
\eps/200.$
Since $\obj(v^\ast_1,\dots,v^\ast_{K-1},w''_K,\dots,w''_n)
\geq \opt - \eps/40$ by Lemma \ref{lem:anytail},
we get that
$\Pr[(v^\ast)^{(H)} \cdot X^{(H)} + R \geq \theta] \geq \opt - 6\eps/200.$

By Lemma \ref{lem:findbestheadworks}, the {\tt Find-Best-Head}
procedure returns a head vector $u^{(H)}=(u_1,\dots,u_{K-1})$
such that
$\Pr[u^{(H)}\cdot X^{(H)} + R \geq \theta]  \geq
\Pr[(v^\ast)^{(H)} + R \geq \theta]$, so
$\Pr[u^{(H)}\cdot X^{(H)} + R \geq \theta]  \geq \opt - 6\eps/200.$
Now recalling that $\dK(R,(w'')^{(T)} \cdot X^{(T)}) \leq \eps/200$,
applying Fact \ref{fact:dK-indep} again gives us that
$\dK(u^{(H)} \cdot X^{(H)} + R,
u^{(H)} \cdot X^{(H)} + (w'')^{(T)} \cdot X^{(T)}) \leq
\eps/200.$
Hence it must be the case that
$\Pr[u^{(H)}\cdot X^{(H)} + (w'')^{(T)} \cdot X^{(T)} \geq \theta]
\geq \opt - 7\eps/200.$
Since $u_1 + \cdots + u_{k-1} + w''_k + \cdots + w''_n \leq 1$
by Lemma \ref{lem:findbestheadworks}, this vector is a near-optimal
feasible solution.
This concludes the proof of Theorem \ref{thm:type-K}, modulo the
proof of Lemma \ref{lem:findbestheadworks}.
\qed


\begin{framed}
\noindent {\bf Algorithm
{\tt Find-Approximately-Best-Head:}}

\medskip

\noindent {\bf Input:}  vector of tail weights $(u_K,\dots,u_n)$
with $u_K + \cdots + u_n \leq 1$;
parameters $\eps',\delta'$

\noindent {\bf Output:}  if $v^\ast$ is type $K$,
with probability $1-\delta'$ a head vector such that
$\Pr[u \cdot X \geq \theta] \geq \Pr[(u'_1,\dots,u'_{K-1},u_K,\dots,u_n)
\cdot X \geq \theta] - \eps'$ for all $(u'_1,\dots,u'_{K-1}) \in \R^{K-1}_{\geq 0}$
such that $u'_1 + \cdots + u'_{K-1} + u_k + \cdots + u_n \leq 1$

\begin{enumerate}

\item Sample $m=\Theta(\log(1/\delta')/(\eps')^2)$ points $t_1,\dots,t_m$
from the random variable $(u_K,\dots,u_n) \cdot X^{(T)}$.
Let $R$ be the random variable which is uniform over the multiset
$\{t_1,\dots,t_m\}.$

\item
Run Algorithm {\tt Find-Best-Head}$(t_1,\dots,t_m,1-\sum_{j=K}^n u_j,K)$
and return the head vector $(u_1,\dots,u_{K-1})$ that it returns.

\end{enumerate}

\end{framed}



\begin{framed}
\noindent {\bf Algorithm {\tt Find-Best-Head:}}

\medskip

\noindent {\bf Input:}  points $t_1,\dots,t_m$,
weight value $0 \leq W \leq 1$, parameter $K$

\noindent {\bf Output:}  Returns the non-negative head vector $u^{(H)} =
(u_1,\dots,u_{K-1})$ that maximizes $\Pr[u^{(H)} \cdot X^{(H)} +
R \geq \theta]$ subject to $u_1 + \cdots + u_{K-1} \leq W$,
where $R$ is the random variable that is uniform over multiset $\{t_1,\dots,t_m\}$

\begin{enumerate}

\item Let ${\cal S}$ be the set of all $2^{\Theta(K^2)}$ sets $S \subseteq \{0,1\}^{K-1}$ such that
$S=\{x \in \{0,1\}^{K-1}:  u \cdot x \geq c \}$ for some $u \in \R^{K-1}, c \in \R.$

\item For each $S=(S_1,\dots,S_m) \in {\cal S}^m$, check whether the following linear program over variables $w_1,\dots,w_{K-1}$
is feasible and if so let $w^{(S)} \in \R^L$ be a feasible solution:

\begin{enumerate}

\item For each $i \in [m]$ and each $x \in S_i$, $w \cdot x + t_i \geq \theta$;

\item $w_1, \cdots, w_{K-1} \geq 0$;

\item $w_1 + \cdots + w_{K-1} \leq W$.

\end{enumerate}

\item For each $w^{(S)}$ obtained in the previous step, compute
$\Pr[w^{(S)} \cdot X^{(H)} + R \geq \theta]$
and output the vector $w^{(S)}$ for which this is the largest.

\end{enumerate}

\end{framed}

\begin{lemma} \label{lem:findbestheadworks}
The (deterministic) algorithm {\tt Find-Best-Head} runs in time $2^{\poly(m,K)}$
and outputs a vector $u^{(H)}=(u_1,\dots,u_{K-1})
\in \R^{K-1}_{\geq 0}$ with $\|u^{(H)}\|_1 \leq W$
which is such that for every $(u')^{(H)} \in \R^{K-1}_{\geq 0}$
with $\|(u')^{(H)}\|_1 \leq W,$ we have
$\Pr[u^{(H)} \cdot X^{(H)} + R] \geq \Pr[(u') \cdot X^{(H)} + R].$
\end{lemma}

\begin{proof}
The claimed running time bound follows easily from the fact that $|{\cal S}|=2^{\Theta(mK^2)}$ (note that
the running time of the linear program and the time required to explicitly compute the probabilities in Step
3 are both dominated by the enumeration over all elements of ${\cal S}^m.$).

The correctness argument is similar to the proof of Theorem \ref{thm:L-junta}.  As in that proof, ${\cal S}$
consists of all possible sets of satisfying assignments to a $(K-1)$-variable halfspace.
The optimal head vector that maximizes
$\Pr[u^{(H)} \cdot X^{(H)} + R \geq \theta]$ subject to $u_1 + \cdots + u_{K-1} \leq W$
must be such that there is some $S=(S_1,\dots,S_m) \in {\cal S}^m$ such that for $1 \leq i \leq m$, $S_i$ is precisely the set of those $x \in \{0,1\}^L$ for which $u^{(H)} \cdot x + t_i \geq \theta.$  By searching over all
$S=(S_1,\dots,S_m) \in {\cal S}^m$ in Step 2, the algorithm will encounter this $S$
and will construct a feasible head vector for it.  Such a feasible head vector will be identified as maximizing the probability in Step 3, and hence {\tt Find-Best-Head} will indeed output an optimal head vector as claimed.
This concludes the proof of Theorem \ref{thm:L-junta}.
\end{proof}


\fi


\section{Putting it together:  proof of Theorem \ref{thm:main}}
\label{sec:together}

In this section we prove Theorem \ref{thm:main} using Theorems
\ref{thm:L-junta}, \ref{thm:type-Lplus1} and \ref{thm:type-K}.

The overall algorithm works as follows.  First, it runs
{\tt Find-Optimal-Junta}$((p_1,\dots,p_L),\theta,1)$
to obtain a feasible solution $w^{\small{\mathrm{junta}}}$.
Next, for each $K=1,\dots,L$
it runs Algorithm {\tt Find-Near-Opt-Small-CI}$((p_1,\dots,p_n),\theta,K,
\delta/(2L))$
to obtain a set ${\cal FEAS}^{(K)}$ of feasible solutions.
Finally, it runs Algorithm {\tt Find-Near-Opt-Large-CI}$((p_1,\dots,p_n),
\theta)$ to obtain a final set ${\cal FEAS}^{(L+1)}$ of feasible
solutions.
It is easy to see from Theorems
\ref{thm:L-junta}, \ref{thm:type-Lplus1} and \ref{thm:type-K} that the running
time of the overall algorithm is as claimed.

Let ${\cal ALL}$ denote the union of the sets
$\{w^{\small{\mathrm{junta}}}\},$ ${\cal FEAS}^{(1)},$ $\dots,$
${\cal FEAS}^{(L)}$ and ${\cal FEAS}^{(L+1)}$.
Since $v^\ast$ must fall in either Case 1, Case 2 or Case 3,
Theorems \ref{thm:L-junta}, \ref{thm:type-Lplus1} and \ref{thm:type-K}
together guarantee that ${\cal ALL}$
is a set of $\poly(n,2^{\poly(L)})$
many feasible solutions that with probability at least $1-\delta/2$
contains a feasible solution $w$ with $\obj(w) \geq \opt - \eps/2.$

Next, we sample $m=\Theta( (1/\eps)^2  \cdot (\log |{\cal ALL}|/\delta))$
points independently from $\calD_p$.  For each feasible solution
$w \in {\cal ALL}$
we use these $m$ points to obtain an empirical estimate
$\widetilde{\obj}(w)$ of $\obj(w)$
(recall that $\obj(w)=\Pr_{X \sim \calD_p}[w \cdot X
\geq \theta]$), i.e., we set $\widetilde{\obj}(w)$ to be the fraction of the
$m$ points that satisfy $w \cdot X \geq \theta.$
A straightforward Chernoff bound implies
that with probability at least $1-\delta/2$, for each $w$ we have
$
|\widetilde{\obj}(w) - \obj(w)| \leq \eps/4.
$

Finally, we output the vector $w^{\ast} \in {\cal ALL}$ that
maximizes $\widetilde{\obj}(w)$
(breaking ties arbitrarily), together with the value
$\widetilde{\obj}(w)$.  With overall probability at least $1-\delta$
this $w^\ast$ has
$\obj(w^\ast) \geq \opt - 3\eps/4$
and $|\widetilde{\obj}(w) - \opt| \leq \eps$ as desired.
This proves Theorem \ref{thm:main}.

\qed

\bibliography{allrefs}
\bibliographystyle{alpha}

\end{document}